\documentclass[acmsmall,nonacm]{acmart} 
\usepackage{graphicx,amsmath,xcolor,amsthm} 
\graphicspath{{figs/}{}}
\usepackage{tikz}

\colorlet{darkred}{red!50!black!100}

\newcommand\BVH[1]{{\color{green}BVH}: {\color{green!50!black!100}\emph{#1}}}
\newcommand\red[1]{{#1}}
\title{Studying the Impact of Service Discipline on Redundancy via Mean Field and Pair Approximations}
\title{The Impact of Service Discipline on Redundancy}
\title{Approximations to Study the Impact of the Service Discipline in Systems with Redundancy}
\author{Nicolas Gast}
\orcid{0000-0001-6884-8698}
\email{nicolas.gast@inria.fr}
\affiliation{
  \institution{Univ. Grenoble Alpes, Inria}
  \city{Grenoble}
  \country{France}
}
\author{Benny Van Houdt}
\orcid{0000-0002-5955-8493}
\email{benny.vanhoudt@uantwerpen.be}
\affiliation{
    \institution{Dept. Computer Science, University of Antwerp}
    \city{Antwerp}
    \country{Belgium}
}
\date{April 2023}

\setcopyright{acmlicensed}
\acmJournal{POMACS}
\acmYear{2024} \acmVolume{8} \acmNumber{1} \acmArticle{14} \acmMonth{3}\acmDOI{10.1145/3639040}
\received{October 2023}
\received[revised]{January 2024}
\received[accepted]{January 2024}

\newtheorem{theorem}{Theorem}

    \ccsdesc[500]{General and reference~Performance}
    \ccsdesc[500]{Mathematics of computing~Queueing theory}
    \ccsdesc[300]{Theory of computation~Scheduling algorithms}

\keywords{load balancing; redundancy; scheduling; queueing theory; mean field approximation; pair approximation}

\begin{document}

\begin{abstract}
As job redundancy has been recognized as an effective means to improve performance of large-scale computer systems, queueing systems with redundancy have been studied by various authors. Existing results include methods to compute the queue length distribution and response time but only when the service discipline is First-Come-First-Served (FCFS). For other service disciplines, such as Processor Sharing (PS), or Last-Come-First-Served (LCFS), only the stability conditions are known. 

In this paper we develop the first methods to approximate the queue length distribution in a queueing system with redundancy under various service disciplines. We focus on a system with exponential job sizes, \emph{i.i.d.} copies, and a large number of servers. We first derive a mean field approximation that is independent of the scheduling policy. In order to study the impact of service discipline, we then derive refinements of this approximation to specific scheduling policies. In the case of Processor Sharing, we provide a {\it pair} and a {\it triplet} approximation.   The pair approximation can be regarded as a refinement of the classic mean field approximation and takes the service discipline into account, while the triplet approximation further refines the pair approximation.  We also develop a pair approximation for three other service disciplines: First-Come-First-Served, Limited Processor Sharing and Last-Come-First-Served. We present numerical evidence that shows that all the approximations presented in the paper are highly accurate, but that none of them are asymptotically exact (as the number of servers goes to infinity). This makes these approximations suitable to study the impact of the service discipline on the queue length distribution. Our results show that FCFS yields the shortest queue length, and that the differences are more substantial at higher loads.
\end{abstract}

\maketitle

\section{Introduction}
Job redundancy has been shown to improve performance in a variety of practical systems \cite{dolly,dean2013,vulimiri2013}.
The main idea is to assign replicas of incoming jobs to multiple servers in the hope that one of the
replicas completes quickly in one of the servers. This has lead to the introduction of several queueing models with redundancy. These models can be classified in a number of ways. A first classification can be made based on the assumptions made on the size of the replicas. In some models replicas are
assumed to have independent and identically distributed (i.i.d.) job sizes \cite{gardner2015reducing,gardner2017redundancy,anton2021stability,BONALD201770}, while other models assume that replicas have identical sizes \cite{anton2021stability,raaijmakers_PS,Hellemans_indentical}. An intermediate model is the S\&X model \cite{gardner2017better}, which assumes that the time required to complete a job on a server is a mixture of a fixed job size and an i.i.d. slowdown of the server.

Another classification can be made based on when the replicas are cancelled. A first option is to
cancel all replicas as soon as one completes \cite{gardner2015reducing,gardner2017redundancy,Hellemans_indentical,anton2021stability,raaijmakers_PS,BONALD201770}, this model is called the cancel-on-complete (CoC) model. For servers using First-Come-First-Served (FCFS) replicas may also be cancelled as soon as one replica starts service \cite{hellemansSIG18,ayesta_unifying}, in which case the model is referred to as the Cancel-on-Start (CoS) model. Other variations such as delayed CoC have also been considered \cite{Hellemans_workload}.
While the stability condition of the CoS model typically demands that the load is bounded by one, 
determining the stability condition of a CoC model is much more challenging \cite{gardner2017redundancy,raaijmakers2020,raaijmakers_PS,anton2021stability}, see
\cite{anton_survey} for a recent survey. 

Two main approaches have been used to derive performance measures of queueing models 
with redundancy. When a finite number of servers is considered \cite{gardner2015reducing,gardner2017redundancy,BONALD201770}, product form
solutions for the queue length distribution under FCFS have been obtained in a number of settings  as these models can be linked either with the order-independent queues in \cite{Krzesinski2011} or with the framework presented in \cite{visschers2012product}. Another approach to study queueing models
with redundancy exists in developing approximations for systems
with a large number of servers (typically a homogeneous system) by using mean field models \cite{gardner2017redundancy,Hellemans_workload,hellemansSIG18,Hellemans_indentical, shneer2020large}.
In this case the workload process is studied, which yields approximation results for
the FCFS response time distribution in a large system. In some particular cases these
approximations are proven to be asymptotically exact \cite{shneer2020large}.

Apart from the line of work focusing on the stability of CoC models, surveyed in \cite{anton_survey}, all prior queueing models with redundancy focused on the performance of FCFS servers. In this paper {\bf we are the first to study the impact of the service discipline on 
the queue length distribution in a 
queueing system with redundancy}. The service disciplines considered are Processor Sharing (\textbf{PS}),
\textbf{FCFS}, Last-Come-First-Served (\textbf{LCFS}) and Limited Processor Sharing $K$ (\textbf{LPS($K$)}), where up to $K$ of the oldest jobs are served in parallel. We consider a CoC queueing model\footnote{In Appendix \ref{apx:JIQ-d}) we further indicate that the model that we study is also relevant
in the context of Join-Idle-Queue load balancing \cite{lu1} with redundant tokens. } with $n$ servers, exponential job sizes and i.i.d. replicas where jobs arrive according to a Poisson process with rate $\lambda n$ and each new incoming job is assigned to $d$ random servers. Throughout the paper we mainly
consider  the case with $d=2$, but most of the presented models 
generalize to the setting with $d > 2$ (see Appendix \ref{apx:d>2}). 
We focus on developing accurate approximations for systems with many servers.

The contributions made in the paper can be summarized as follows. First we introduce a mean field
model  that is not affected by the service discipline and that can be used to approximate the queue length distribution in a server. This model basically assumes that queues that hold replicas of the same
job have independent queue lengths. While this is true when the replicas are placed in the servers,
this assumption becomes false afterwards. This is intuitively clear, if a replica of a job has spent a long time period in a long queue, the other replica did not complete service during that period and
is therefore more likely to be part of a long queue as well.

Second, we introduce an exact dynamic graph model when the service discipline is PS. However solving
this graph model seems to be intractable. We therefore introduce a so-called {\it pair approximation}
that relies on the assumption that the graph model becomes a local tree when the number of
nodes in the graph becomes large and that there exists a conditional independence between the degree
of two nodes that are both connected to the same node with a given degree.
We show  by simulation that this pair approximation is highly accurate for systems with many servers. 
We provide numerical evidence that shows that the pair approximation is not asymptotically exact as the number of servers goes to infinity. 

Our third contribution is the introduction of  a
further refinement of the pair approximation for PS servers, called the {\it triplet approximation} that uses a more involved conditional
independence assumption on the degree of the  nodes in the graph. Simulation shows that
this approximation is even more accurate than the pair approximation. Despite its remarkable
accuracy the triplet approximation is also asymptotically inexact and coming up with an approximation
that captures the exact limiting behavior for PS servers appears to be hard.

The fourth contribution of the paper exists in applying the idea of the pair approximation
for PS servers to other service disciplines. More specifically, we show how to develop
pair approximations for systems with FCFS, LCFS and LPS($K$) servers.
The main distinction between the initial mean field model and these pair approximations is
that the pair approximations take some of the correlations between the queue lengths of servers that share a job into account. We also show that if we assume independence between queue lengths of replicas, all the pair 
approximations collapse to the initial mean field model.
As such all our pair approximations can be regarded as refinements of the original mean
field model, tailored such that they take the specifics of the queueing discipline into account.

Finally we compare numerically the performance of the different service disciplines. The results show that FCFS performs best in terms of the queue length
distribution and thus in terms of the mean response time with gains exceeding
$20\%$ compared to LCFS under high loads.

The paper is structured as follows. Section~\ref{sec:model} introduces the model and shows that
the service discipline matters. Section~\ref{sec:approx} presents the mean field approximation, which is insensitive to the service discipline.
Section~\ref{sec:PS-approx} contains the pair and triplet approximations for PS servers. Section~\ref{sec:service-disciplines} shows how to derive pair approximations for FCFS, LPS($K$) and LCFS. Section~\ref{sec:numerical} presents numerical results on the impact of the service discipline.
Section~\ref{sec:conclusions} contains conclusion and future work.

\section{Model}
\label{sec:model}

\subsection{The redundancy-d model}

We consider a system of $n$ identical servers where $n$ is large. Jobs arrive in the system at rate $n\lambda$. When a job arrives, $d$ replicas  of this jobs are created and sent to $d$ servers chosen at random. We consider a cancel-on-complete policy: this means that when a replica of a job is completed, the $d-1$ other replicas are canceled and removed from the other servers. This occurs instantaneously.  We refer to these $d-1$ other replicas of a job as its $d-1$ {\it buddies}.  We assume that the size of all replicas are \emph{i.i.d} and exponentially distributed with mean $1$.  While \emph{i.i.d.}
and exponentially distributed job sizes is somewhat restrictive from a practical
point of view, we explain in Appendix~\ref{apx:JIQ-d} that this restricted setting
is of particular interest to study Join-Idle-Queue load balancing with redundancy (JIQ-RED), which was recently introduced in \cite{Harte}.

\subsection{Service discipline matters}
\label{ssec:PSneFCFS}

It is shown in \cite{anton2021stability} that the stability region of PS and FCFS is identical, that is, both models are stable for $\lambda < 1$ under exponential job sizes when the mean job size equals $1$. 
\red{Further, in \cite{antonMAMA2023} the authors illustrate a coupling approach  that suggests that the queue length distribution is also identical for a variety of service disciplines, including PS and FCFS, in a redundancy-$d$ system.} However, we provide numerical evidence in this section
that illustrates that the queue length distribution of a server \emph{is} affected by the service
discipline even for a model of exponentially distributed and i.i.d.~replicas.
We illustrate this in two settings: (i) in a system with $n=3$ servers, where we present
both numerical and simulation results, and (ii) in a system with $n=10^5$ servers using simulation.

For the $3$ server case we numerically compute the steady-state distribution for the redundancy-d model with $d=2$. For FCFS, an explicit formula can be derived from the product form distribution of \cite{gardner2017redundancy}, and is obtained by using the order-independent-queue framework. For PS, we construct a Markov chain with $(K+1)^3$ states, where $K$ is the maximal number of jobs per class. In such a chain, the state is represented as $(s_1,s_2,s_3)$ where $s_i$ is the number of jobs for which there is a replica in servers $i$ and $i+1$ (with the convention that $3+1=1$).  For each $i$, $s_i$ becomes $s_i+1$ at rate $\lambda$ (arrival) and $s_i$ becomes $s_i-1$ at rate $s_i/(s_i+s_{i+1})+s_i/(s_i+s_{i-1})$ because there are $s_i+s_{i-1}$ replicas currently served in server $i$ (with the convention $s_0=s_3$).  We obtain the steady-state distribution by inverting the matrix. 
The steady-state distribution of the number of jobs in a server for PS and for FCFS are reported in Table~\ref{tab:steady-state}.  This table shows that the two steady-state distribution are close, but distinct.  To rule out numerical issues or effects due to the bounded value of $K$, we compare the theoretical values with simulations that show the exact same trend. 

\begin{table}[t]
    \begin{tabular}{|c|c|c|c|c|c|c|c|c|}
        \hline
        Number of jobs & mean & 1 & 2 & 3 & 4 & 5 & 6& 7\\
        \hline
        FCFS (exact \cite{gardner2017redundancy}) & 0.88889 & 0.27500 & 0.12875 & 0.05619 & 0.02365 & 0.00975 & 0.00397 & 0.00161 \\
        FCFS (simu) & 0.88892 & 0.27499 & 0.12875 & 0.05619 & 0.02365 & 0.00975 & 0.00397 & 0.00161 \\\hline
        PS (chain $K{=}20$) & 0.89225 & 0.27423 & 0.12862 & 0.05641 & 0.02391 & 0.00994 & 0.00409 & 0.00167 \\
        PS (simu) & 0.89217 & 0.27422 & 0.12862 & 0.05642 & 0.02391 & 0.00994 & 0.00409 & 0.00167 \\\hline
    \end{tabular}
    \caption{Mean queue length and steady-state distribution of the queue lengths of servers for a redundancy-2 system with $n=3$, arrival rate $\lambda=0.5$ and service rate $\mu=1$.}
    \label{tab:steady-state}
\end{table}

For the setting with $n=10^5$ servers we simulated a system with $d=2$ for FCFS, LPS($K$), PS and LCFS,. Table~\ref{tab:simul10000} presents simulation results for $\lambda=0.9$ averaged over four runs of length $10^5$.  While we observed in Table~\ref{tab:steady-state} that the differences between the different service disciplines are quite small at low loads, the differences are much more pronounced at high load. Our simulation results clearly show differences in queue length distribution that cannot be attributed to simulation noise.

\begin{table}[t]
    \centering
    \begin{tabular}{|c|c|c|c|c|c|c|c|c|}
    \hline 
        Number of jobs & mean & 1 & 2 & 3 & 5 & 10\\ \hline
         FCFS:	&3.1171  & 1.4884e-01  & 1.7984e-01 &  1.7882e-01 &  1.0707e-01 &  3.3898e-03\\\hline
         LPS(2): & 3.2243  & 1.4480e-01  & 1.7229e-01   &1.7197e-01 &  1.0919e-01 &  4.8019e-03\\
         LPS(5):& 3.3483 &1.4140e-01  & 1.6511e-01 &  1.6434e-01   &1.0958e-01  & 7.0372e-03\\
         LPS(10): &3.3884 &  1.4030e-01  & 1.6296e-01 &  1.6211e-01 &  1.0959e-01  & 7.7873e-03 \\
         \hline
         PS:&	3.3890  & 1.4029e-01  & 1.6293e-01 &  1.6207e-01 &  1.0959e-01  & 7.8041e-03\\ \hline
         LCFS:& 3.9415 &  1.2816e-01 &  1.3965e-01&   1.3677e-01  & 1.0405e-01  & 1.9717e-02
         \\\hline
    \end{tabular}
    \caption{Simulation results for number of jobs in a queue of a redundancy-2 system with $n=100000$ servers, arrival rate $\lambda=0.9$ and service rate $\mu=1$. }
    \label{tab:simul10000}
\end{table}

\section{Mean field approximation}
\label{sec:approx}

In this section, we construct a mean field approximation for the redundancy-d model. This approximation does not depend on the service discipline as long as it is work-conserving. As we will see later, this approximation is not asymptotically exact, and one can derive more accurate approximations by using pair or triplet approximations.

\subsection{Derivation of the ODE: transient regime}
To construct the mean field approximation, let us first compute the rate at which a replica is served. Let $q_t(x)$ be the fraction of servers who have $x$ replicas at time $t$ and pick a replica at random (uniformly among all replicas in the system at time $t$). This replica is in a server with $y$ replicas with probability $y q_t(y)/\bar{q}_t$, where $\bar{q}_t = \sum_{y\ge1} yq_t(y)$. This implies that a random replica is served at rate 
\begin{align*}
    \sum_{y\ge1}\frac{1}{y} \frac{y q_t(y)}{\bar{q}_t} = \frac{1-q_t(0)}{\bar{q}_t},
\end{align*}
where in the above equation the factor $\frac{1}{y}$ is due to work-conserving nature of the policy: each server serves replicas at rate $1$, which means that a randomly chosen replica among these $y$ is served at rate $1/y$. For instance, using PS all $y$ replicas
are served at rate $1/y$, hence also a random one. Under FCFS or LCFS the random replica is served with probability $1/y$ at rate $1$, under limited-PS the replicas in the first $K$ positions are served at rate $1/K$ (if $y \geq K$) and the random replica is
served at rate $1/K$ with probability $K/y$. 

Let us now focus on a given server that has $x$ replicas in its queue. There are three types of events that affect this server: \begin{enumerate}
    \item A replica can be added to its queue, which occurs at rate $d\lambda$.
    \item The server might complete the service of one of its replicas, which occurs at rate $1$.
    \item One of the replica can be removed from the queue because another replica of the same job, that is, a buddy, has been served by another server. The rate at which this last event occurs is complex and depends on the queue lengths of all servers who hold the buddies of the $x$ jobs that our current server has. This is where we need an approximation. The mean field approximation consists in 
    {\bf assuming that the queue lengths of the servers holding the different replicas of the same job are independent. 
    This is equivalent to stating that the $d-1$ buddies are random replicas.} This means that each of the $d-1$ buddies disappears 
    at rate $(1-q_t(0))/\bar{q}_t$. Hence, the total rate at which this event occurs is $(d-1)x(1-q_t(0))/\bar{q}_t$.
\end{enumerate}
To summarize, replicas appear at rate $d\lambda$ and disappear at rate $(1+(d-1)x(1-q_t(0))/\bar{q}_t)$ in a server that has $x$ replicas.  This leads to  a system of differential equations for $q_t(x)$: for $x>0$:
\begin{align}\label{eq:indepq}
    \frac{d q_t(x)}{dt} =~ & d\lambda \left[q_t(x-1)-q_t(x)\right] + q_t(x+1) - q_t(x) \\
    &+  \frac{(1-q_t(0))(d-1)}{\bar q_t} \left[(x+1)q_t(x+1)-xq_t(x)) \right],\nonumber
\end{align}
where $q_t(0)=1-\sum_{x\geq 1} q_t(x)$.

The solution of the ODE~\eqref{eq:indepq} is what we call the mean field approximation. In this equation, the first line corresponds to the local arrival and service of replicas (type 1 and type 2 events), while second line is caused by the completion of buddies in other servers.  We emphasize that this model is just an {\bf approximation}: as we will see later in Section~\ref{sec:numerical}, it is not asymptotically exact as the number of servers grows, contrary to what happens in many systems. This is due to the fact that the queue lengths of servers holding a replica of the same job are not independent (when the number of servers is very large, the $d$ queue lengths are independent at the time when the job arrives, but a dependence is created as time progresses). It can be shown that this implies that the model for JIQ-Red($d$) in \cite{Harte} is not asymptotically
exact either.

\subsection{Fixed point analysis}

The mean field equation leads to an almost-closed form result for the fixed point.
\begin{theorem}\label{th:q}
    The set of ODEs given by \eqref{eq:indepq} has a unique fixed point given by
    \begin{align}\label{eq:mf-fixed}
        q(x) = q(x-1) \frac{d\lambda}{1+\frac{\lambda x(d-1)}{\bar q}} =  (1-\lambda) \frac{(d\lambda)^x}{\prod_{\ell = 1}^x (1+\frac{\lambda \ell(d-1)}{\bar q})}, 
    \end{align}
    where $\bar q$ is the unique solution on $(0,\infty)$ of the equation 
    \[
    \frac{\lambda}{1-\lambda}= \sum_{x\geq 1} \frac{(d\lambda)^x}{\prod_{\ell=1}^x (1+\frac{\lambda (d-1) \ell}{\bar q})} =
    \left(\frac{d \bar q}{d-1}\right)^{-\bar q/(\lambda(d-1))} e^{d\bar q/(d-1)} \gamma(1+\bar q/(\lambda (d-1)),d\bar q/(d-1)),\]
where $\gamma(s,x) = \int_0^x t^{s-1} e^{-t} dt$ is the lower incomplete gamma function
and where the right-hand side is increasing in $\bar q$.
\end{theorem}

\begin{proof}
    Assume $q(0)$ and $\bar{q}$ are fixed, then the ODE \eqref{eq:indepq} essentially corresponds to a birth-death process. Hence, summing \eqref{eq:indepq} for $x \geq z$ implies
    \begin{align*}
        \sum_{x \geq z} \frac{d q_t(x)}{dt}  = -q_t(z) + d \lambda q(z-1) - (d-1)(1-q_t(0))\frac{z q_t(z)}{\bar q_t}.
    \end{align*}
    By setting $\frac{d q_t(x)}{dt}=0$, this implies that a fixed point of \eqref{eq:indepq} must satisfy (for $x>0$):
    \begin{align}\label{eq:sum_indep0}
        q(x) = q(x-1) \frac{d\lambda}{1 + (1-q(0))(d-1)\frac{x}{\bar{q}}}.
    \end{align}
    
    To compute $q(0)$, let us sum $x\frac{d q_t(x)}{dt}$ over $x$. This gives: 
    \begin{align}\label{eq:sum_indep}
        \sum_{x\ge1} x \frac{d q_t(x)}{dt}  
        & =  d \lambda - \sum_x q_t(x) + \frac{1-q_t(0)}{\bar q_t} \sum_x x q_t(x) \nonumber \\
        &= \lambda - 1+q_0(t).
    \end{align}
    This shows that any fixed point of the ODE~\eqref{eq:indepq} satisfies $q(0) = 1 - \lambda$ and therefore implies \eqref{eq:mf-fixed}.
    
    As $\sum_{x\geq 1} q(x)=\lambda$, $\bar q$ is a solution of
    \begin{align}\label{eq:fixed_barq}
        \frac{\lambda}{1-\lambda} =   \sum_{x \geq 1} \frac{(d\lambda)^x}{\prod_{\ell = 1}^x (1+\frac{\lambda \ell(d-1)}{\bar q})}. 
    \end{align}
    The right-hand side of this equality is clearly increasing in $\bar q$ and equals $0$ in $0$. We may therefore
    conclude that $\bar q$ is the unique fixed point of \eqref{eq:fixed_barq}.
    Remark that
    \[ \sum_{x\geq 1} \frac{a^x}{\prod_{\ell=1}^x (1+b \ell)} = 
    \left(\frac{a}{b}\right)^{-1/b} e^{a/b} \gamma(1+a/b,a/b),\]
    where $\gamma(s,x) = \int_0^x t^{s-1} e^{-t} dt$ is the lower incomplete gamma function.
    Letting $a=\lambda d$ and $b=\lambda (d-1)/\bar q$ completes the proof.
\end{proof}

\section{RED-PS: Dynamic graph model and pair/triplet approximations}
\label{sec:PS-approx}

What makes the mean field approximation asymptotically inexact is the dependencies between the servers: if a job has one replica in a server with $x$ replicas and its buddy is located in a server with $y$ replicas, then this job is more likely to disappear if $y$ is small than if $y$ is large. This implies that the queue lengths of servers holding replicas of the same job are not independent
(notwithstanding the fact that they were independent when the job was just added).  In this section, we take this dependence into account for the redundancy system with PS service by looking at the graph of dependence between servers.  
Other service disciplines are discussed in Section \ref{sec:service-disciplines}. For simplicity and numerical tractability, we restrict our attention to $d=2$ replicas. 
In Appendix \ref{apx:d>2} we indicate how our models can be generalized to $d > 2$ replicas.

\subsection{The dynamic graph model}

When there are $d=2$ replicas per job, the RED-PS model can be represented by a dynamic random (multi)graph $(V, E_t)$. The graph is undirected and the set of vertices $V$ is the set of servers. The set of edges $E_t$ evolves over time and there is an edge between $u \in V$ and $v \in V$ for each job that is shared by servers $u$ and $v$ at time $t$.  The graph $(V,E_t)$ is in fact a multigraph since there can be multiple edges between two nodes when two servers share more than one job. There can be also loops if two replicas of the same jobs are connected to the same server.  An example of such a graph is depicted in Figure~\ref{fig:dynamic_graph}(a). 

\begin{figure}[ht]
    \centering
    \begin{tabular}{ccc}
        \begin{minipage}{0.28\linewidth}
            \centering
            \begin{tikzpicture}
                \tikzset{every loop/.style={}}
                \tikzstyle{vertex}=[circle,fill,inner sep=0pt,minimum size=0.1cm]
                \newcommand\vertex[4]{\node[vertex] at (#1,#2) (#3)  {} ;%
                \node[above=0.2] at (#1,#2) {\small #4};}
                \vertex{0}{0}{0}{2};
                \vertex{1}{0}{1}{4};
                \vertex{1.7}{-.5}{2}{2};
                \vertex{1.7}{.5}{3}{2};
                \vertex{2.4}{.5}{4}{2};
                \vertex{3.1}{.5}{5}{1};
                \vertex{0}{-.7}{6}{0};
                \vertex{.5}{-.7}{7}{0};
                \vertex{2.4}{-.7}{8}{1};
                \vertex{3.1}{-.7}{9}{1};
                \draw (0) -- (1) -- (3) -- (4) -- (5) (8) -- (9);
                \draw (1) edge[bend left] (2);
                \draw (1) edge[bend right] (2);
                \draw (0) edge[loop left] (0);
            \end{tikzpicture}
        \end{minipage}
        &
        \begin{minipage}{0.28\linewidth}
            \centering
            \begin{tikzpicture}
                \tikzset{every loop/.style={}}
                \tikzstyle{vertex}=[circle,fill,inner sep=0pt,minimum size=0.1cm]
                \newcommand\vertex[4]{\node[vertex] at (#1,#2) (#3)  {} ;%
                \node[above=0.2] at (#1,#2) {\small #4};}
                \vertex{0}{0}{0}{1};
                \vertex{1}{0}{1}{3};
                \vertex{1.7}{-.5}{2}{1};
                \vertex{1.7}{.5}{3}{2};
                \vertex{2.4}{.5}{4}{2};
                \vertex{3.1}{.5}{5}{1};
                \vertex{0}{-.7}{6}{0};
                \vertex{.5}{-.7}{7}{0};
                \vertex{2.4}{-.7}{8}{1};
                \vertex{3.1}{-.7}{9}{1};
                \draw (0) -- (1) -- (2) (1) -- (3) -- (4) -- (5) (8) -- (9);
            \end{tikzpicture}
        \end{minipage}
        &\begin{minipage}{.35\linewidth}
            \begin{tabular}{|c|c|c|c|c|c|}
            \hline
            $(x,y)$ & (1,1) & (1,2) & (1,3) & (2,3)\\
            & (2,2) & (2,1) & (3,1) & (3,2)\\\hline
            $\pi(x,y)$ & $\frac1{10}$ & $\frac1{20}$ & $\frac{1}{10}$& $\frac{1}{20}$ \\
            \hline
        \end{tabular}
        \end{minipage}\\~\\
        (a) The multigraph model
        &(b) The (simple) graph
        &(c) Values $\pi(x,y)$
    \end{tabular}    
    \caption{Example of a dynamic random graph for a system with $n=10$ servers. The labels on the servers indicate degrees. A new edge is added between two nodes at rate $2\lambda$. An edge connecting two nodes $(u,v)$ of degrees $d_t(u)$ and $d_t(v)$ disapears at rate $1/d_t(u)+1/d_t(v)$.}
    \label{fig:dynamic_graph}
\end{figure}
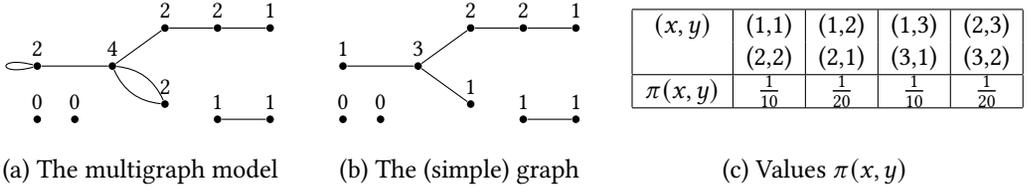

The multigraph evolves as follows:
\begin{itemize}
    \item For any pair of servers $(u,v)$, a new edge $(u,v)$ is added to the graph at rate $2\lambda/n$. 
    \item Any vertex has an internal Poisson clock of intensity $1$. When this clock ticks, it destroys one of its edges (taken uniformly at random among all its edges).  
\end{itemize}
Let us denote by $d_t(u)$ the degree of a node $u$ at time $t$, which is also equal to the number of replicas that this server holds. The above transitions imply that if an edge connects two servers of degree $d_t(u)$ and $d_t(v)$, then this edge disappears at rate $1/d_t(u)$ + $1/d_t(v)$.

The multigraph model corresponds exactly to the RED-PS model. In what follows, we will use this graph representation to construct two sets of ODEs that are a very good approximation of the RED-PS model when the number of servers $n$ is large. Yet, to simplify the derivation of the equations, we will construct these ODEs by using a slightly simpler model in which we forbid the creation of multi-edges or loops -- the rest of the model remains the same.  The argument to do so is that, when the number of servers $n$ is large, the rate at which loops or multi-edges are created is $O(1/n)$, while the rate at which such edges are removed remains $O(1)$. Hence, the difference between the two models is likely to be of order $O(1/n)$.

\subsection{The Pair Approximation}\label{sec:PSpair}

To construct our mean field approximation in Section \ref{sec:approx}, the idea was to focus on servers and study the number of replicas per server. To obtain a more accurate approximation, here we change our point of view and focus on the jobs rather than on servers, and we look at how many replicas do the two servers of a job hold. 

\subsubsection{Fraction of pairs}
As before, we denote by $d_t(u)$ the degree of a node $u$ at time $t$. It is equal to the number of neighbors in the graph $G_t$, that is, $d_t(u) = |\Gamma_t(u)|$, where $\Gamma_t(u)=\{v \mid (u,v) \in E_t\}$ is the set of neighbors of the node $u$.  We denote by $\pi_t(x,y)$ the number of edges that connect a node of degree $x$ with a node of degree $y$ at time $t$ divided by $n$ (the number of servers) if $x=y$, and by $2n$
if $x \not=y$:
\begin{align}
    \label{eq:pi(x,y)}
    \pi_t(x,y) &= \frac{1}{2n} \sum_{u \in V} \sum_{v \in \Gamma_t(u)} 1[d_t(u)=x, d_t(v)=y],
\end{align}
where the factor $1/2$ is because one can exchange the role of $u$ and $v$: an edge that connects two nodes with a different degree, say degree $x$ and $y$, contributes $1/2n$ to $\pi_t(x,y)$ and $1/2n$ to $\pi_t(y,x)$ whereas an edge that connects two nodes of degree $x$ simply contributes $1/n$ to $\pi_t(x,x)$.

By construction, the quantity $\pi_t(x,y)$ is symmetric and only makes sense if $x\ge1$ and $y\ge1$ because if a server is connected to a job, it has at least one replica. For $x\ge1$, we denote:
\begin{align}
    \label{eq:pi(x)}
    \pi_t(x) = \sum_{y \geq 1} \pi_t(x,y) = \sum_{y \geq 1} \pi_t(y,x).
\end{align}
As for the mean field approximation, we denote by $q_t(x)$ the fraction of servers holding exactly $x$ jobs. The total number of replicas residing in servers with queue length $x$ is equal to $2n\pi_t(x)$. Notice that here we count every edge $(x,y)$ 
with $x \not= y$ once (as it contributes $1/2n$ to $\pi_t(x,y)$) and every edge $(x,x)$ twice as this edge corresponds to
two such replicas. Dividing this total number $2n\pi_t(x)$ by $nx$ yields $q_t(x)$ (for $x\ge1$):
\begin{align}
    \label{eq:q_from_pi}
    q_t(x) = \frac{2\pi_t(x)}{x}.
\end{align}
We end by defining the conditional probability that a random job in a queue of length $x$ has a buddy in a queue with length $y$,
which we denote as $\pi_t(y|x)$.  As stated above the total number of replicas residing in servers with queue length $x$ is equal to $2n\pi_t(x)$, while $2n \pi_t(x,y)$ is the total number of replicas in a queue with length $x$ connected to a queue with length $y$.
Hence,
\begin{align}\label{eq:condpi}
    \pi_t(y|x) = \frac{\pi_t(x,y)}{\pi_t(x)},
\end{align}
which holds irrespective of whether $x$ equals $y$.

We illustrate the definition of $\pi$ in Figure~\ref{fig:dynamic_graph}(c) where we show the values of $\pi(x,y)$ for the graph of Figure~\ref{fig:dynamic_graph}(b). 
One can verify that the degree distribution is $q(1)=\pi(1)=2(1/20+1/10+1/10)=5/10$, $q(2)=2\pi(2)/2=(1/10+1/20+1/20)=2/10$, $q(3)=2(1/10+1/20)/3=1/10$ and $q(0)=1-q(1)-q(2)-q(3)=2/10$. For the conditional probabilities, we also have for instance that $\pi(2|1)=\pi(1,2)/(\pi(1,1)+\pi(1,2)+\pi(1,3))=1/5$, and $\pi(1|2)=1/4$. Note that conditional probabilities are (in general) not symmetric.

\subsubsection{Rates}

In what follows, we call $(x,y)$ the type of a job. There are two main types of events that can modify the number of jobs of a given type.  First, there can be a completion or a creation of a job: at rate $1/x+1/y$ a job of type $(x,y)$ is served, in which case we have one job less of type $(x,y)$. A new job is created at rate $n\lambda$, in which case it connects to two servers at random. These servers have $x-1$ and $y-1$ jobs with probability $2q_t(x-1)q_t(y-1)$ if $x\not= y$ and $q_t(x-1)^2$ otherwise.  Notice that when $x \not= y$  the edge that is created contributes $1/2n$ to $\pi_t(x,y)$ and $1/2n$ to $\pi_t(y,x)$ so the $2$
vanishes. Hence, there is a creation of a job of type $(x,y)$ at rate $n\lambda q_t(x-1)q_t(y-1)$.

Second, jobs can also change type.  The type can change from $(x,y)$ to $(x+1,y)$ (or to $(x,y+1)$) if a replica is added to the first server (or to the second). This occurs at rate $2\lambda$. The type can also change from $(x,y)$ to $(x-1,y)$ in two cases:
\begin{enumerate}
    \item When any of the other $x-1$ replicas is completed in the first server. This occurs at rate $(x-1)/x$. 
    \item When a job is cancelled because the buddy of one of the $x-1$ replicas in the first server completes. Here, we need 
    to make an approximation. The pair-approximation consists in {\bf assuming that if two jobs happen to share a single server $s$, 
     the queue lengths of the two other servers (other than $s$) that also serve these jobs have independent queue lengths given the length of $s$.}  Mathematically speaking the pair approximation consists in assuming that $\pi_t(y'|x)$ is the probability that the buddy of a job in a server of length $x$ has length $y'$.
     This corresponds to assuming that the server length of the buddy of any of these $x-1$ jobs only depends on $x$
     (and not on $y$).  
    As a result, a buddy of the $x-1$ other replicas completes at rate $\sum_{y' \geq 1}q_t(y'|x) \frac{1}{y'}$, where $1/y'$ is the completion rate of a replica when there are $y'$ jobs under PS.
\end{enumerate}
Summing both terms, a job transitions from type $(x,y)$ to $(x-1,y)$ at rate $\frac{x-1}{x} + h_t(x)$, where 
\begin{align}
    \label{eq:h-pair}
    h_t(x) = (x-1) \sum_{y' \geq 1} \pi_t(y'|x) \frac{1}{y'}.
\end{align}
Similarly, a job transitions from $(x,y)$ to $(x,y-1)$ at rate $\frac{y-1}{y} +h_t(y)$.

Note that the rate at which the job transitions from $(x,y)$ to $(x-1,y)$ does not depend on $y$. This is because we approximate the rate at which buddies are served by \eqref{eq:h-pair}. This equation is where we make an approximation, by considering that the queue length of the buddies of the replicas in a server of length $x$ only depends on $x$, but not on the other buddies. This is called a \emph{pair} approximation because we limit the graph of dependence to pairs.

\subsubsection{The set of ODEs for the pair approximation}

Based on the above transitions, we can now write an ODE that indicates how the quantities $\pi_t(x,y)$ evolve. For $x,y\ge1$, this ODE is: 
\begin{align}\label{eq:ODEPS}
    \frac{d \pi_t(x,y)}{dt} &= \lambda q_t(x-1) q_t(y-1) 
    +2\lambda \left[\pi_t(x-1,y)+\pi_t(x,y-1)-2\pi_t(x,y)\right] 
    \nonumber \\ 
    &\quad+ \pi_t(x+1,y)\left[ h_t(x+1)+ \frac{x}{x+1}\right] + \pi_t(x,y+1)\left[ h_t(y+1)+ \frac{y}{y+1}\right]  \nonumber\\
    &\quad - \pi_t(x,y)\left[2+h_t(x)+h_t(y)\right],
\end{align}
where $q_t$ is defined as in \eqref{eq:pi(x)} and \eqref{eq:q_from_pi}, and with the convention that $\pi_t(0,y)=\pi_t(x,0)=0$ for all $x,y\ge0$. 

Equation~\eqref{eq:ODEPS} can be understood as follows. The first term corresponds to the creation of jobs of type $(x,y)$, and the second term corresponds to job changing types because of arrivals. The second line corresponds to jobs changing from type $(x+1,y)$ or $(x,y+1)$ to type $(x,y)$. The third line is actually a sum of two transition types: completions of jobs of type $(x,y)$ (which would be a term $\pi_t(x,y)(1/x+1/y)$), plus jobs of type $(x,y)$ transitioning to $(x-1,y)$ or $(x,y-1)$ (which would be a term $\pi_t(x,y)\left[(x-1)/x+h_t(x)+(y-1)/y+h_t(y)\right]$). Summing both terms lead to the last line of \eqref{eq:ODEPS}. In Appendix~\ref{apx:morePS}, we preset additional results for the pair approximation concerning the properties of the fixed point and the relationship between the pair approximation and the mean field approximation.

\subsubsection{Model Validation}

In this section we present some numerical results to illustrate the accuracy of the pair approximation
and the original mean field model using simulation. It also includes numbers for the triplet approximation
for later reference, which we ignore for now.
We simulated systems with $n=10^2, 10^3, 10^4$ and $10^5$
servers and loads $\lambda=0.5, 0.7$ and $0.9$.  The simulation numbers were averaged over $4$ runs of length $10^5$ (or more for smaller $n$) each. Table \ref{tab:simulPAIR09} presents the results for $\lambda=0.9$.
The numbers for $\lambda = 0.5$ and $0.7$ can be found in the Appendices in Tables \ref{tab:simulPAIR05} and \ref{tab:simulPAIR07}. The first column in each table presents the mean queue length, the second to last column
present the probability that the queue length equals $1,2,3,5$ and $10$, respectively. The fixed point of the
set of ODEs characterizing the pair approximation was computed by using a simple Euler iteration starting
from an empty system at time $t=0$.

We first note that the pair approximation is much more accurate than
the mean field approximation, with relative errors below $0.1\%$ except for the last probability.
We also remark that first the five digits of the simulation numbers are identical for $n=10^4$ and $n=10^5$, meaning that increasing $n$ further probably has limited impact. This suggests that convergence has occurred and the pair
approximation is asymptotically inexact. 
\begin{table}[t]
    \centering
    \begin{tabular}{|r|c|c|c|c|c|c|c|c|}
    \hline 
Number of jobs & mean & 1 & 2 & 3 & 5 & 10\\ \hline
\hline 
Simul $n=10^2$ &3.4449  & 1.3943e-01 &  1.6088e-01  & 1.5932e-01  & 1.0855e-01  & 9.1793e-03 \\
 $n=10^3$&   3.3947 &  1.4018e-01 &  1.6272e-01 &  1.6178e-01 &  1.0949e-01  & 7.9446e-03\\
 $n=10^4$ &   3.3890 &  1.4030e-01 &  1.6293e-01  & 1.6205e-01 &  1.0956e-01  & 7.8142e-03\\
 $n=10^5$&   3.3890 &  1.4030e-01 &  1.6293e-01  & 1.6205e-01 &  1.0956e-01  & 7.8142e-03\\
\hline
Triplet Approx. & 3.3894 &   1.4020e-01  & 1.6289e-01 &  1.6207e-01 &  1.0962e-01 &  7.8059e-03 \\\hline
Pair Approx. &   3.3875 &  1.4032e-01 &  1.6300e-01  & 1.6215e-01 &  1.0960e-01  & 7.7756e-03\\
\hline
Mean field &   3.3377 &  1.4177e-01 &  1.6578e-01  & 1.6496e-01 &  1.0950e-01  & 6.8438e-03\\
         \hline
    \end{tabular}
    \caption{Accuracy of Mean Field, Pair and Triplet Approximation for PS servers, arrival rate $\lambda=0.9$. }
    \label{tab:simulPAIR09}
\end{table}

\subsection{Triplet approximation}\label{sec:triplet}
In this section we introduce the so-called {\it triplet} approximation, which is meant to further improve upon the pair approximation. While the pair approximation was obtained by considering pairs of servers that share jobs, here we go beyond, by consider triplets of servers such that the first two share a job and the last two share another job. This requires more formalism to construct the approximation.  As for the pair approximation, we consider the case with two replicas ($d=2$).

\subsubsection{Connected triplets}

Recall from \eqref{eq:pi(x,y)} that $\pi_t(x,y)$ is the number of pairs of servers with degree $x$ and $y$ rescaled by the number of servers $n$ (or $2n$ if $x \not= y$). 
In this part, we now focus on triplets.  Consider three nodes $u,v,w \in V$ such that $(u,v) \in E_t$, $(v,w)\in E_t$, $d_t(u)=x$, $d_t(v)=y$ and $d_t(w)=z$. We refer to these three nodes as a triplet of type $(x,y,z)$ or $(z,y,x)$, and we denote by $c_t(x,y,z)$ the number of triplets of type
$(x,y,z)$, rescaled by $n$ (or $2n$ if $x \not= z$). For all $x,y,z$, they can be expressed as:
\begin{align}
    c_t(x,y,z) & = \frac{1}{2n}\sum_{v \in V} \sum_{u \in \Gamma_t(v)} \sum_{w \in \Gamma_t(v), w \not= u}  1[d_t(u)=x,d_t(v)=y,d_t(w)=z],
\end{align}
where as before $\Gamma_t(u)$ are the neighbors of $u$. 

The factor $1/2$ in the above equation is because the triplet can be read in two directions when $x \not= z$ : $(x,y,z)$ and $(z,y,x)$. We therefore let every such triplet contribute $1/(2n)$ to $c_t(x,y,z)$ and $1/(2n)$ to $c_t(z,y,x)$ if $x \not= z$. When $x=z$ it simply contributes $1/n$ to $c_t(x,y,z)$.  Note that the model is symmetric in the sense that $c_t(x,y,z)=c_t(z,y,x)$. By construction, $c_t(x,y,z)$ is zero if $y=1$ or if any of $x$, $y$ or $z$ equals $0$. For example, if we look at Figure \ref{fig:dynamic_graph}(b) we find that
$c_t(1,2,2)=c_t(2,2,1)=1/20$ and $c_t(1,3,2)=c_t(2,3,1)=1/10$.

To build our approximation, we will need the degree distribution and the conditional degree of a node given its neighbors. As shown in \eqref{eq:q_from_pi}, the degree can be expressed from $\pi_t(x,y)$ for $x\ge1$.  Similarly, when $y\ge2$, the quantity $\pi_t(x,y)$ can be expressed from the quantities $c_t(x,y,z)$ by noting that each triplet $(x,y,z)$ corresponds to $y-1$ pairs $(x,y)$. By using \eqref{eq:q_from_pi}, we can also recover the degree from $c_t(x,y,z)$ for all $y\ge2$: 
\begin{align*}
    \pi_t(x,y) &= \frac{1}{y-1}\sum_{z \geq 1} c_t(x,y,z),\\
    q_t(y) &= \frac{2}{y} \sum_{x \geq 1} \pi_t(x,y) = \frac{2}{y(y-1)} \sum_{x,z \geq 1} c_t(x,y,z).
\end{align*}
This last expression is also immediate because every degree $y$ node has $y(y-1)/2$ triplets with $y$ as a middle node. 

When $y=1$, the node $v$ cannot be in the middle of a triplet. It can be connected to a node of degree $1$ (in which case they form a pair but not a triplet), or it can be connected to a node of degree $x\ge2$ which is again connected to a node $z\ge1$. When $x\ge2$, there are $x-1$ triplets that have the degree one node $v$ as an end node. This implies that 
\begin{align*}
    q_t(1) = 2\sum_{x\ge1} \pi_t(1,x) = 2 \pi_t(1,1) + 2 \sum_{x \geq 2,z \geq 1} \frac{c_t(1,x,z)}{x-1}.
\end{align*}

We also need the conditional probability, denoted as $c_t(x|y)$ that a random neighbor of a random degree $y$ node has degree $x$
at time $t$ (which was denoted as $\pi_t(x|y)$ in the pair approximation model). Assume $y\geq 2$, then we can count the scaled number of triplets with a degree $y$ middle node and a degree $x$ end node,
which equals $2\sum_z c_t(x,y,z) / (y-1)$ as every such triplet is counted $y-1$ times, and
divide by the scaled number of neighbors of all degree $y$ nodes, given by $yq_t(y)$.
Therefore, for $y\ge2$ and $x\ge1$, we have:
\begin{align}
    c_t(x|y) = \frac{2}{y(y-1)q_t(y)}\sum_{z \geq 1} c_t(x,y,z), 
\end{align}
where the $2$ is due to the symmetry.  Note that the same expression also follows from \eqref{eq:condpi}. For a degree $1$ node we can use the same reasoning to find
\begin{align}
    c_t(x|1) &= \frac{2}{(x-1)q_t(1)}\sum_{y \geq 1} c_t(1,x,y)\qquad \text{ for $x\ge2$,}\\
    c_t(1|1) &= \frac{2\pi_t(1,1)}{q_t(1)}.
\end{align}
It can be verified from the formulas that, by construction, $\pi_t(y|x)=c_t(y|x)$.

The last building block that we need to define is the conditional probability that the degree of an end node
of a triplet equals $z$ given that its other end node has degree $x$ and the middle node has degree $y$
at time $t$, which
we denote as $c_t(z|x,y)$. Due to symmetry we have
\begin{align}
    c_t(z|x,y) = \frac{c_t(x,y,z)}{\sum_{z'} c_t(x,y,z')}.
\end{align}

\subsubsection{Rates}
We first look at the evolution of $\pi_t(1,1)$, the number of edges between two degree one nodes divided by $n$. Such an edge disappears whenever there is either a service completion or an arrival in either of its nodes. Service completions occur at rate $1$ in each server, while arrivals occur at rate $2\lambda$ in each node. This results in a total death rate of $2+4\lambda$. Creation of pairs $(1,1)$, with an edge between two degree one nodes, can be caused by two types of events: 
\begin{itemize}
    \item First, there could be an arrival where both jobs are assigned to an idle server. This event occurs at rate $\lambda q_t(0)^2$.
    \item Second, it is also possible that a triplet $(u,v,w)$ is such that $d_t(u)=1$ and $d_t(v)=2$ and that the edge that connects $v$ to its other neighbor, say node $w$, disappears due to a service completion. The rate at which this edge $(v,w)$ disappears due to a service completion depends on $d(w)$. If $d(w)=x$, then this rate equals $1/d_t(u) + 1/d_t(w) = 1/2+1/x$.
\end{itemize} 
This yields a total creation rate of pairs $(1,1)$ equal to
\begin{align}
    \lambda q_t(0)^2 + 2 \sum_{x\ge1} c_t(1,2,x) \left( \frac{1}{2} + \frac{1}{x}\right),
\end{align}
where the factor $2$ is due to symmetry.

We now turn our attention to the evolution of $c_t(x,y,z)$ over time. 

(1*) A new triplet is created if a new arrival links two nodes (one of which has a degree larger than $0$). There are two ways to create an $(x,y,z)$ triplet: (1) a degree $x-1$ node connects to a degree $y-1$ node that has a degree $z$ neighbor, or (2) a degree $z-1$ node connects to a degree $y-1$ node that has a degree $x$ neighbor. For case (1) the new arrivals must occur in a degree $x-1$ and $y-1$ node. Moreover we demand that the server selected for the first replica has degree $x-1$ and the server selected for the second replica has degree $y-1$. If the two job replicas are placed in the other order, we state that a type $(z,y,x)$ triplet is created. In this manner the symmetry in the system is preserved. As the degree $y$ node had $y-1$ neighbors and each of these has degree $z$ with probability $c(z|y-1)$ the rate at which a type $(x,y,z)$ triplet is created is given by
\[
\lambda q_t(x-1) q_t(y-1) c(z|y-1) (y-1),
\]
in case (1). For case (2) we state that a type $(x,y,z)$ triplet is created when replica one goes to the degree $y-1$ server and the second to the degree $z-1$ server (otherwise we state that it is a type $(z,y,x)$ triplet). For this case, the rate a which we create type $(x,y,z)$ triplets is therefore given by
\[
\lambda q_t(y-1) q_t(z-1) c(x|y-1) (y-1).
\]

Consider now an existing triplet $(x,y,z)$. 
(2*) There could be an arrival in either one of the three nodes (at rate $2\lambda$). This simply changes the type to $(x+1,y,z)$,
$(x,y+1,z)$ or $(x,y,z+1)$ depending on where the arrival occurs. Arrivals in servers that do not belong to the
triplet $(x,y,z)$ and that share a job with one of the three servers do not affect  a triplet $(x,y,z)$.

(3*) There could be a service completion in either of the three servers (at rate $1$ in each server). 
Assume a service completion occurs
in the first node (with degree $x$), then with probability $1/x$ the job that is shared with the middle node
completes service and the triplet $(x,y,z)$ disappears. Otherwise, with probability $(x-1)/x$, the triplet
changes from type $(x,y,z)$ to $(x-1,y,z)$. The same holds for the other end node. For the middle node
there is a probability of $2/y$ that a job shared with either one of the endpoints completes, making 
the triplet $(x,y,z)$ disappear. Or with probability $(y-2)/y$ the triplet changes into a type $(x,y-1,z)$
triplet.

(4*) Service completions in a server $s$ not belonging to the triplet $(x,y,z)$ that share a job with one of the nodes part of the
triplet do however affect the type of the triplet $(x,y,z)$. Let us first consider the case where $s$ is a neighbor of an end
point of the triplet, say the node with degree $x$ (for $z$ the argument is similar). In this case the degree $v$ of $s$ is
computed as $c_t(v|y,x)$. Notice that {\bf this degree is only an approximation as we assume that it is independent of
the degree of the other endpoint, that is, it is assumed independent of $z$.} As 
the degree $x$ node had $x-1$ neighboring nodes
that are not part of the triplet, the rate of such service completions equals 
\[ \sum_{v \geq 1} c_t(v|y,x) \frac{x-1}{v}.\]
In this case the type $(x,y,z)$ of the triplet becomes $(x-1,y,z)$. The same happens when a service completion
occurs in a neighbor of the endpoint with degree $z$ different from $y$ and the rate is given by
$ \sum_v c_t(v|y,z) \frac{z-1}{v}$, 
and the type becomes $(x,y,z-1)$.

A more difficult case arises when $s$ is a neighbor of the middle node with degree $y$. In order to determine the
probability that $s$ has degree $v$, we should use the fact that it is connected to the middle node of an $(x,y,z)$
triplet. However, in order to define our triplet approximation we only have triplet information.
{\bf We therefore need to define an approximation.} One option would
be to neglect one of the endpoints. For instance we could neglect $z$ and use $c(v|x,y)$. Or we could neglect
the end node with the smallest or largest degree. After experimenting with different such choices, we decided to 
determine the degree $v$ of $s$ as follows. We  count all the triplets of the form $(x,y,v)$ and $(v,y,z)$ and
divide by all the triplets that have a degree $x$ and $y$ nodes as neighbors. In other words, we use $c_t(v|x,y,z)$
which we define as
\begin{align}
    c(v|x,y,z) = \frac{c_t(x,y,v)+c_t(v,y,z)}{\sum_{v'} (c_t(x,y,v')+c_t(v',y,z)) }.
\end{align}
As the middle node has $y-2$ neighbors not belonging to the triplet, the rate is given by
\[ \sum_v c_t(v|x,y,z) \frac{y-2}{v}.\]

\subsubsection{The set of ODEs}
For ease of notation let $c_t(x,y,z)=0$ of $x<1$, $y<2$ or $z<1$.
Given our discussion of the rates in the previous subsubsection, we have
\begin{align}
    \label{eq:triplet_11}
    \frac{d}{dt} \pi_t(1,1) = \lambda q_t(0)^2 -(2+4\lambda) \pi_t(1,1)+  2 \sum_x c_t(1,2,x) \left( \frac{1}{2} + \frac{1}{x}\right),
\end{align}
and
\begin{align}\label{eq:ODE_triplet}
    \frac{d}{dt} c_t(x,y,z) = ~&\lambda q_t(x-1)q_t(y-1)c_t(z|y-1)(y-1) + \lambda q_t(y-1)q_t(z-1)c(x|y-1) (y-1) \nonumber\\
    & -6\lambda c_t(x,y,z) + 2\lambda c_t(x-1,y,z) + 2\lambda c_t(x,y-1,z) +2\lambda c_t(x,y,z-1)\nonumber \\
    &-3 c_t(x,y,z) + c_t(x+1,y,z)\frac{x}{x+1}+ c_t(x,y+1,z)\frac{y-1}{y+1} + c_t(x,y,z+1)\frac{z}{z+1} \nonumber\\
    &-c_t(x,y,z) \left( \sum_v c_t(v|y,x) \frac{x-1}{v} + \sum_v c_t(v|y,z) \frac{z-1}{v} + \sum_v c_t(v|x,y,z) \frac{y-2}{v} \right)\nonumber \\
    &+c_t(x+1,y,z) \sum_v c_t(v|y,x+1) \frac{x}{v}  
    \nonumber \\
    & +c_t(x,y,z+1) \sum_v c_t(v|y,z+1) \frac{z}{v}  +c_t(x,y+1,z) \sum_v c_t(v|x,y+1,z) \frac{y-1}{v} 
\end{align}
where the first two terms are due to (1*), the next four terms are due to case (2*), the next $4$ due to (3*), and the last $4$ due to (4*).

We show in Appendix~\ref{apx:pair from triplet} that one can recover the pair approximation from the triplet approximation, which suggests that the triplet approximation can be regarded as a refinement of the pair approximation that takes more correlations  into account.

\subsubsection{Model Validation} 

In this section we illustrate the accuracy of the triplet approximation. We refer back to Tables \ref{tab:simulPAIR09}, \ref{tab:simulPAIR05} and \ref{tab:simulPAIR07} for the fixed point values obtained by the triplet approximation. The results confirm that the triplet approximation can be regarded as a further refinement of the pair approximation (though not every probability is necessarily closer).  We also illustrate this improvement using Figure \ref{fig:non-asympto-exact} that 
visualizes that the triplet approximation is asymptotically more accurate than the pair approximation. These results suggest that the pair approximation is asymptotically inexact. Similarly we believe that the triplet approximation is also asymptotically inexact and can probably be further refined, for instance using a {\it quadruplet} approximation.  Note that the gain in accuracy of the triplet approximation is, however, minor and as such the pair approximation might suffice to gain insights on the impact of the service discipline.

\begin{figure}[t]
    \centering
    \begin{tabular}{@{}c@{}c@{}c@{}}
        \includegraphics[width=0.33\linewidth]{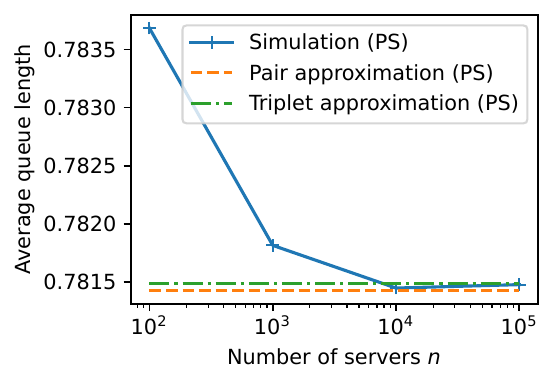}
        &\includegraphics[width=0.33\linewidth]{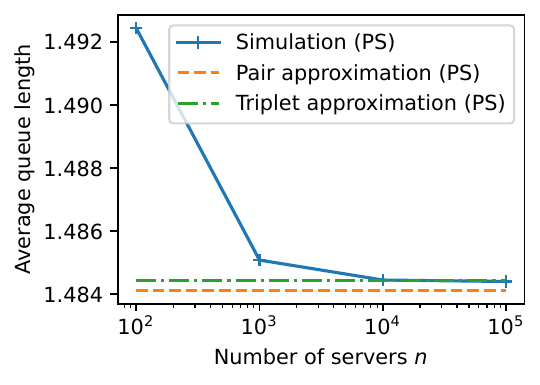}
        &\includegraphics[width=0.32\linewidth]{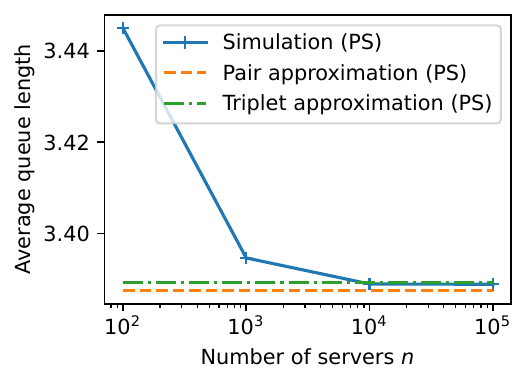}
        \\
        (a) PS, $\lambda=0.5$
        &(b) PS, $\lambda=0.7$
        & (c) PS, $\lambda=0.9$
    \end{tabular}
    \caption{Average queue length as a function of the number of servers $n$: we compare the numbers obtained by simulation (for finite $n$) to the pair and triplet approximations.}
    \label{fig:non-asympto-exact}
\end{figure}

\section{Extension to other service disciplines}
\label{sec:service-disciplines}

We now turn our attention to three other service disciplines: First-Come-First-Served (FCFS), Limited Processor Sharing (LPS) and Last-Come-First-Served (LCFS). In each case, we explain how to construct a pair approximation. The key difference with processor sharing is that one also needs to keep track of the position of the jobs in a queue as the queue length
alone does not suffice.  These approximations collapse to the same mean field approximation when an independence assumption is made.  As before we consider a system with $d=2$,
but these models can be generalized to $d>2$ in a similar fashion as the generalization for
PS servers discussed in Appendix \ref{apx:d>2}. 

\subsection{First-Come-First-Served (FCFS)}\label{sec:FCFS}

\subsubsection{Types of jobs}

Each job is composed of two replicas that we call the first and the second replica. We denote by $2n\delta_t(x_1,y_1,x_2,y_2)$ the number of jobs that have their first replica  in position $x_1$ 
in a queue with length $x_2$ and their second replica in position $y_1$ in a queue with length $y_2$
(or vice versa). Note that $x_1 \leq x_2$ and $y_1 \leq y_2$.
As before, the quantity $\delta_t$ is symmetric (\emph{i.e.}, $\delta_t(x_1,y_1,x_2,y_2)=\delta_t(y_1,x_1,y_2,x_2)$ for all $x_1,x_2,y_1,y_2\ge1$) and we denote by 
$\delta_t(x_1,y_1,x_2)=\sum_{y_2\ge y_1}\delta_t(x_1,y_1,x_2,y_2)$,  $\delta_t(x_1,x_2)=\sum_{y_1\ge1}\delta_t(x_1,y_1,x_2)$, $\delta_t^{(pos)}(x_1)=\sum_{x_2} \delta_t(x_1,x_2)$ and by  $\delta_t^{(ql)}(x_2)=\sum_{x_1} \delta_t(x_1,x_2)$. 

We once more denote by $q_t(x)$ the fraction of the servers holding exactly $x$ replicas. There are different ways to relate $q_t$ and $\delta_t$. In total, there are $2n\delta_t^{(pos)}(x)$ replicas that are residing in position $x$ of a server. Hence, we have: 
\begin{align}\label{eq:qxFCFS}
    q_t(x) = 2\delta_t^{(pos)}(x)- 2\delta_t^{(pos)}(x+1),
\end{align} 
as a queue has length $x$ if there is a replica in position $x$, but no replica in position $x+1$.
An alternate, equivalent way to define $q_t(x)$ based on $\delta_t^{(ql)}(x)$ is to state that
\begin{align}\label{eq:qxFCFS2}
    q_t(x) = \frac{2 \delta_t^{(ql)}(x)}{x},
\end{align} 
as in the PS pair approximation (see \eqref{eq:q_from_pi}).
Finally note that $q_t(x)$ can also be expressed as $2\delta_t(1,x)$ as there is exactly one replica
in position $1$ of any queue with length $x$.

\subsubsection{Rates}

We now look at the evolution of number of jobs of type $(x_1,y_1,x_2,y_2)$. Each job of type $(1,y_1,x_2,y_2)$ or $(x_1,1,x_2,y_2)$ is served at rate $1$, hence we have a death rate of type $(1,y_1,x_2,y_2)$ and $(x_1,1,x_2,y_2)$ jobs equal to $1$.  A new type $(x_1,y_1,x_2,y_2)$ job is created at rate $\lambda q_t(x_2-1)q_t(y_2-1)$
if $x_1=x_2$ and $y_1=y_2$ as any new job joins the back of the queue. This is the birth rate of type $(x_2,y_2,x_2,y_2)$ jobs.

Jobs can also change type.  When using FCFS, a job type can never increase $x_1$ or $y_1$ because future arrivals join at the back of the queue. Future arrivals do increase $x_2$ and $y_2$ as follows.
As any arrival picks two servers at random, any server is subject to an arrival rate of $2\lambda$.
This implies that jobs of type $(x_1,y_1,x_2,y_2)$ change into jobs of type
$(x_1,y_1,x_2+1,y_2)$ at rate $2\lambda$ and similarly into jobs of type $(x_1,y_1,x_2,y_2+1)$
at the same rate of $2 \lambda$.

The state of a job can also change due to service completions.
We distinguish two cases: First, the job at the head of the line can complete service, which occurs at rate $1$. In this case both the position and queue length decrease by one.
Second, a buddy of another job that shares a queue with our tagged job may complete service in some other server. Here we must further differentiate between two possibilities: it can be
the buddy of a job ahead of the tagged job (so the buddy of a job in positions $1$ to $x_1-1$
in the first queue or positions $1$ to $y_1-1$ in the second queue) or a buddy of a job that is 
positioned after the tagged job. In the former case again both the position and queue length
decrease by one, but in the latter case only the queue length decreases by one.

To compute these rate as a function of $\delta_t(x_1,y_1,x_2,y_2)$, we need an \textbf{approximation}. Here, we need the probability that the buddy of the job in position $x'$ of the current queue with length $x_2$ is in position $1$ of its queue. We assume this probability
only depends on $x'$ and $x_2$ and is equal to 
\[\frac{\delta_t(x',1,x_2)}{\delta_t(x',x_2)}\] 
and does not depend on $x_1,y_1$ or $y_2$. In other words we assume this probability only depends on the
position and queue length. 
Summing the two possibilities that decrease both $x_1$ and $x_2$ implies that the rate
of change from $(x_1,y_1,x_2,y_2)$ for $x_1 > 1$ to $((x_1-1,y_1,x_2-1,y_2)$ is equal to 
\begin{align*}
    1+ \sum_{x' =1}^{x_1-1} \frac{\delta_t(x',1,x_2)}{\delta_t(x',x_2)},
\end{align*}
for all $x_1\ge2$, while the case that only decreases $x_2$ by one has rate
\begin{align*}
\sum_{x' =x_1+1}^{x_2} \frac{\delta_t(x',1,x_2)}{\delta_t(x',x_2)},
\end{align*}
For further use let 
$\kappa_t(x_a,x_b)= \sum_{x'=1}^{x_a} \frac{\delta_t(x',1,x_b)}{\delta(x',x_b)}$. 
The above two rates can be denoted as $1+\kappa_t(x_1-1,x_2)$ and
$\kappa_t(x_2,x_2)-\kappa_t(x_1,x_2)$.
The situation is symmetric for $y$.

\subsubsection{The set of ODEs}
The above rates lead to the following set of ODEs for $\delta_t(x,y)$, where
$\delta_t(x_1,y_1,x_2,y_2)=0$ if one of its entries equals zero or $x_1 > x_2$ or $y_1 > y_2$:
\begin{align}\label{eq:ODEFCFS}
    \frac{d}{dt} \delta_t(x_1,y_1,&x_2,y_2)= \lambda q_t(x_2-1) q_t(y_2-1) 
    1[x_1=x_2, y_1=y_2]- 4\lambda \delta_t(x_1,y_1,x_2,y_2) \nonumber \\&+
    2\lambda\delta_t(x_1,y_1,x_2-1,y_2) + 2\lambda \delta_t(x_1,y_1,x_2,y_2-1) \nonumber\\
    &-\delta_t(x_1,y_1,x_2,y_2) 
    \left[2+ \kappa_t(x_2,x_2) - \frac{\delta_t(x_1,1,x_2)}{\delta_t(x_1,x_2)}
    + \kappa_t(y_2,y_2) - \frac{\delta_t(y_1,1,y_2)}{\delta_t(y_1,y_2)}\right] 
    \nonumber \\ 
    &+\delta_t(x_1+1,y_1,x_2+1,y_2) 
    \left[1+ \kappa_t(x_1,x_2+1) \right] \nonumber \\
    &+\delta_t(x_1,y_1,x_2+1,y_2) 
    \left[\kappa_t(x_2+1,x_2+1) -\kappa_t(x_1,x_2+1) \right] \nonumber \\
    & +\delta_t(x_1,y_1+1,x_2,y_2+1) 
    \left[1+ \kappa_t(y_1,y_2+1) \right] \nonumber \\
    &+\delta_t(x_1,y_1,x_2,y_2+1) 
    \left[\kappa_t(y_2+1,y_2+1) -\kappa_t(y_1,y_2+1) \right].
\end{align}

\subsubsection{Model Validation}

Tables \ref{tab:simulPAIR09FCFS}, \ref{tab:simulPAIR05FCFS} and \ref{tab:simulPAIR07FCFS} correspond to
Tables \ref{tab:simulPAIR09}, \ref{tab:simulPAIR05} and \ref{tab:simulPAIR07}, except that the PS service
discipline has been replaced by FCFS servers. \red{An additional line has also
been inserted that correspond to $E[T] = (\ln (1/(1-\lambda)) - \lambda)/\lambda^2$ which is an approximation of the mean queue length for Redundancy-$2$ with FCFS servers presented in \cite{gardner2017redundancy} and shown to be asymptotically exact in \cite{shneer2020large}.}
These tables demonstrate that the pair approximation
for FCFS servers drastically reduces the error of the mean field approximation and yields errors
below $0.1\%$, except for the last probability. They do not provide the value for $E[T]$ as \cite{gardner2017redundancy} and therefore are not asymptotically exact. The advantage of our method compared to \cite{gardner2017redundancy} is to provide an approximation of the full queue length distribution.

\begin{table}[t]
    \centering
    \begin{tabular}{|r|c|c|c|c|c|c|c|c|}
    \hline 
        Number of jobs & mean & 1 & 2 & 3 & 5 & 10\\ \hline
         \hline
         Simul $n=10^2$&    3.1793 &  1.4732e-01&   1.7636e-01 &  1.7481e-01 &  1.0716e-01   &4.4893e-03\\
 $n=10^3$&  3.1227   &1.4872e-01 &  1.7950e-01  & 1.7840e-01 &  1.0709e-01  & 3.4895e-03\\
    $n=10^4$&3.1171   &1.4884e-01 &  1.7983e-01  & 1.7879e-01 &  1.0707e-01  & 3.3952e-03\\
   $n=10^5$& 3.1171   &1.4884e-01 &  1.7983e-01  & 1.7879e-01 &  1.0707e-01  & 3.3952e-03\\\hline
  \red{2$\lambda E[T]$ \cite{gardner2017redundancy}}&  \red{3.1169}   & - &  -  & -&  -  & - \\ \hline
   Pair Approx & 3.0965   &1.4921e-01 &  1.8097e-01  & 1.8031e-01 &  1.0706e-01  & 3.0484e-03\\\hline
  Mean Field &  3.3377   &1.4177e-01 &  1.6578e-01  & 1.6496e-01 &  1.0950e-01  & 6.8438e-03\\
         \hline
    \end{tabular}
    \caption{Accuracy of Mean Field and Pair Approximation for FCFS servers, arrival rate $\lambda=0.9$. }
    \label{tab:simulPAIR09FCFS}
\end{table}

\subsection{Limited Processor Sharing (LPS)}

In this subsection we generalize the FCFS model to a model for limited processor sharing LPS($K$) where $K$ is the
maximum number of jobs that are served in parallel by a single server. \red{LPS(K) works as follows: When there are $K$ or more
$K$ jobs in the queue, the oldest $K$ are served in parallel (in a PS fashion), otherwise they are all served in parallel (as for PS).} The model reduces to the FCFS model of Section~\ref{sec:FCFS} when $K=1$, and to the pair approximation for PS of Section \ref{sec:PSpair} when $K=\infty$.

\subsubsection{Types of jobs and rates}

The type of a job is represented as $(x_1,y_1,x_2,y_2)$, where $x_2$ and $y_2$ are the queue lengths
of the servers holding the two replicas for the job as in the FCFS model. However, $x_1$ now equals
$1$ if the first replica is being served (meaning it is in one of the first $K$ positions),
while $x_1 > 1$ if the replica is in position $K+x_1-1$, which means it is not in service.
The same holds for $y_1$ and replica $2$. Remark that when $x_1 > 1$ (or $y_2 > 1$), 
then $x_2 > K$ (or $y_2>K$).
If we set $K=\infty$ then both $x_1$ and $y_1$
always equal one and the type of a job is simply characterized by the queue lengths $(x_2,y_2)$ as for the PS pair approximation. 

Let $\mu_t(x_1,y_1,x_2,y_2)$ be the number of type $(x_1,y_1,x_2,y_2)$ jobs scaled by the number of
servers at time $t$. Define $\mu_t(x_1,x_2) = \sum_{y_1,y_2} \mu_t(x_1,y_1,x,y_2)$
and $\mu_t(x)=\sum_{x_1} \mu_t(x_1,x)$.
The queue length distribution can be expressed in terms of the $\mu_t(x)$ 
\begin{align}\label{eq:qxLPS}
    q_t(x) = \frac{2 \mu_t(x)}{x},
\end{align} 
using the same reasoning as in \eqref{eq:q_from_pi} and \eqref{eq:qxFCFS2}.

New arrivals increase the $x_2$ or $y_2$ value of a job by one, where each server is subject to new arrivals at rate
$2\lambda$. A new job is also created when an arrival occurs. The difference with the
FCFS model is that the $x_1$ and $y_1$ values do not match $x_2$ and $y_2$, respectively,
when a new job is created. Let $[x]^+ = \max(0,x)$, then $x_1$ equals $1+[x_2-K]^+$
if a new job joins a queue with length $x_2$ (such that $x_1$ equals $1$ if $x_2 \leq K$). Similarly $y_1$ equals $1+[y_2-K]^+$. This yields the first two lines in \eqref{eq:ODELPS}.

Service completions occur in both servers at rate $1$. Assume the service completion
occurs in the server holding the first replica. If $x_1=1$ the tagged job completes service
with probability $p^K_1(x_2)=\min((x_2-1)/x_2,(K-1)/K)$, otherwise the tagged job only changes
type (that is, $x_1$ and $x_2$ decrease by one). This explains lines $3$ and $4$ in \eqref{eq:ODELPS}.

Job types can also change due to a service completion of a buddy of a job that shares a 
server with the tagged job. To determine the rate at which these service completions occur
we need to introduce an {\bf approximation}. We focus on the change in $x_1$ and $x_2$ as the
same holds for $y_1$ and $y_2$ due to the symmetry of the model. Consider the buddy of the job
that is not our tagged job.  Let $x'= 1$ if this job is in one of the first $K$ positions 
and let $x'>1$ if this job is in position $K+x'-1$. The approximation exists in assuming that the probability that this buddy is in service in a server with length $y'$ is given by 
$\mu_t(x',1,x_2,y')/\mu_t(x',x_2)$. In other words, the probability that the buddy of the job  is in service in a queue of length $y'$ depends only on $x_2$ and the value of $x'$ (which is fully determined by its
position). If a buddy is in one of the first $K$ positions of a server with queue length $y'$,
then it is served at rate $1/\min(K,y')$. Hence the completion rate of the buddy of $x'$ is
\begin{align}\label{eq:buddyrate}
    r_t^K(x',x_2) = \sum_{y'} \frac{\mu_t(x',1,x_2,y')}{\mu_t(x',x_2) \min(K,y')}.
\end{align} When $x_2 \leq K$ the total rate at which buddies of jobs other that the tagged job complete
service is thus $\min(K-1,x_2-1)r_t^K(1,x_2)$. When $x_2$ exceeds $K$, we find
that this rate is given by
\[ K r_t^K(1,x_2) + \sum_{x'=2}^{x_2-K+1} r_t^K(x',x_2) - r_t^K(x_1,x_2),\]
where the tagged job is of type $(x_1,y_1,x_2,y_2)$.
Similar to the $\kappa_t(x_a,x_b)$ values for the FCFS model, we now define the $\psi_t(x_a,x_b)$
values, that is,
\begin{align}
    \psi^K_t(1,x_b) &= \min(K,x_b) r_t^K(1,x_b) \\
     \psi^K_t(x_a,x_b) &= K r_t^K(1,x_b) + \sum_{x'=2}^{x_a} r_t^K(x',x_b). 
\end{align}
Using this notation, the total rate at which buddies of the jobs other than the tagged job
complete service can be stated as
$\psi^K_t(1,x_2) - r_t^K(1,x_2)$ when $x_2 \leq K$ and $\psi^K_t(x_2-K+1,x_2)-r_t(x_1,x_2)$
otherwise. This explains line $5$ in \eqref{eq:ODELPS}.

When such a buddy completes service, $x_2$ always decreases by one.
However the value of $x_1$ only decreases by one if the value of $x'$ is strictly less
than $x_1$ (which can only occur when the tagged job is not in service). 
In other words, the total rate at which $x_2$ and $x_1$ both decrease by one is given by
$\psi^K_t(x_1-1,x_2)$, yielding line $6$ in \eqref{eq:ODELPS}.
Finally, when $x_1=1$ or $x' > x_1$ only $x_2$ decreases by one. For $x_1=1$ the total rate
at which  $x_2$ decreases by one due to a service completion of a buddy
can be written as $\psi_t^K(1+[x_2-K]^+,x_2)-r_t^K(1,x_2)$. 
For $x_1 > 1$ only the buddies of the job after position $K+x_1-1$ decrease $x_2$ by
one without decreasing $x_1$, meaning we have a total rate given by
$\psi_t^K(x_2-K+1,x_2)-\psi_t^K(x_1,x_2)$. This yields the last four lines in \eqref{eq:ODELPS}.

\subsubsection{The set of ODEs}
The above rates lead to the following set of ODEs for $\mu_t(x_1,y_1,x_2,y_2)$, where
$\mu_t(x_1,y_1,x_2,y_2)=0$ if one of its entries equals zero or $x_1 > x_2 - K+1$ or $y_1 > y_2-K+1$:
\begin{align}\label{eq:ODELPS}
    \frac{d}{dt} \mu_t(&x_1,y_1,x_2,y_2)= \lambda q_t(x_2-1) q_t(y_2-1) 
    1[x_1=1+[x_2-K]^+, y_1=1+[y_2-K]^+]\nonumber \\&
    - 4\lambda \mu_t(x_1,y_1,x_2,y_2) +2\lambda\mu_t(x_1,y_1,x_2-1,y_2) + 2\lambda \mu_t(x_1,y_1,x_2,y_2-1) \nonumber\\
    &-2\mu_t(x_1,y_1,x_2,y_2) +\mu_t(x_1+1,y_1,x_2+1,y_2) + \mu_t(x_1,y_1+1,x_2,y_2+1)
    \nonumber \\
    &+1[x_1=1]\mu_t(1,y_1,x_2+1,y_2) p_1^K(x_2+1) 
    +1[y_1=1]\mu_t(x_1,1,x_2,y_2+1) p_1^K(y_2+1) \nonumber \\ 
    &-\mu_t(x_1,y_1,x_2,y_2) 
    \left[\psi^K_t(1+[x_2-K]^+,x_2) - r_t^K(x_1,x_2) + \psi^K_t(1+[y_2-K]^+,y_2) - r_t^K(y_1,y_2)\right] 
    \nonumber \\ 
    & +\mu_t(x_1+1,y_1,x_2+1,y_2) \psi^K_t(x_1,x_2+1) +
     +\mu_t(x_1,y_1+1,x_2,y_2+1) \psi^K_t(y_1,y_2+1) 
  \nonumber \\
    &+1[x_1=1] \mu_t(1,y_1,x_2+1,y_2) (\psi_t^K(1+[x_2-K+1]^+,x_2+1)-r_t^K(1,x_2+1))
    \nonumber \\
    &+1[y_1=1] \mu_t(x_1,1,x_2,y_2+1) (\psi_t^K(1+[y_2-K+1]^+,y_2+1)-r_t^K(1,y_2+1))
    \nonumber \\
    &+1[x_1 > 1]\mu_t(x_1,y_1,x_2+1,y_2) 
    \left[\psi^K_t(x_2-K+2,x_2+1) -\psi^K_t(x_1,x_2+1) \right] \nonumber \\
    &+1[x_1 > 1]\mu_t(x_1,y_1,x_2,y_2+1) 
    \left[\psi^K_t(y_2-K+2,y_2+1) -\psi^K_t(y_1,y_2+1) \right].
\end{align}
As with the FCFS pair approximation (see Appendix \ref{apx:moreFCFS}) it is possible to show that $q(0)$ equals $1-\lambda$ in any fixed
point and that this ODE collapses to \eqref{eq:indepq} when using an independence assumption.

\subsubsection{Model Validation}
Tables \ref{tab:simulPAIR09LPS_K2}, \ref{tab:simulPAIR05LPS_K2} and \ref{tab:simulPAIR07LPS_K2} correspond to
Tables \ref{tab:simulPAIR09}, \ref{tab:simulPAIR05} and \ref{tab:simulPAIR07}, where the PS servers
are replaced by LPS(K) servers with $K=2$. These tables demonstrate that the pair approximation
also works nicely for LPS($2$) servers. We obtained similar results for larger $K$ values, that is,
for $K=3,5$ and $10$ (not shown in the paper). These additional experiments confirm that as 
$K$ increases the queue length distribution of the pair approximation for LPS($K$) converges towards the distribution of the PS pair approximation.

\begin{table}[t]
    \centering
    \begin{tabular}{|r|c|c|c|c|c|c|c|c|}
    \hline 
        Number of jobs & mean & 1 & 2 & 3 & 5 & 10\\ \hline
         \hline
Simul $n=10^2$&3.2880   &1.4309e-01  & 1.6921e-01  & 1.6859e-01  & 1.0905e-01  & 6.0308e-03\\
$n=10^3$&   3.2256 &  1.4491e-01&   1.7221e-01 &  1.7171e-01&   1.0903e-01&   4.8872e-03\\
$n=10^4$&   3.2243 &  1.4478e-01&   1.7226e-01 &  1.7194e-01&   1.0919e-01&   4.8114e-03\\
$n=10^5$&   3.2243 &  1.4478e-01&   1.7226e-01 &  1.7194e-01&   1.0919e-01&   4.8114e-03\\\hline
Pair Approx.&   3.2166 &  1.4484e-01&   1.7258e-01 &  1.7248e-01&   1.0932e-01&   4.6412e-03\\\hline
Mean Field&   3.3377 &  1.4177e-01&   1.6578e-01&   1.6496e-01&   1.0950e-01&   6.8438e-03\\
         \hline
    \end{tabular}
    \caption{Accuracy of Mean Field and Pair Approximation for LPS(2) servers, arrival rate $\lambda=0.9$. }
    \label{tab:simulPAIR09LPS_K2}
\end{table}

\subsection{Last-Come-First-Served (LCFS)}

The last service discipline that we consider is last-come-first-served (LCFS).

\subsubsection{Types of jobs and rates}

The analysis for LCFS proceeds similarly to the one of FCFS,. We now use the value $\gamma_t(x_1,y_1,x_2,y_2)$ that is equivalent to the value $\delta_t(x_1,y_1,x_2,y_2)$ for FCFS except that position $1$ now corresponds to the back of the queue. As for FCFS, we denote by $\gamma_t(x_1,y_1,x_2)=\sum_{y_2\ge y_1}\gamma_t(x_1,y_1,x_2,y_2)$,  $\gamma_t(x_1,x_2)=\sum_{y_1\ge1}\gamma_t(x_1,y_1,x_2)$, $\gamma_t^{(pos)}(x_1)=\sum_{x_2} \gamma_t(x_1,x_2)$ and by  $\gamma_t^{(ql)}(x_2)=\sum_{x_1} \gamma_t(x_1,x_2)$. 
The queue length can be expressed as
$q_t(x)=2(\gamma_t^{(pos)}(x)-\gamma_t^{(pos)}(x+1))$
or equivalently as $q_t(x)=2\gamma_t^{(ql)}(x)/x$ or $q_t(x)=2\gamma_t(1,x)$.

At rate $1$ a job of type $(1,y_1,x_2,y_2)$ or $(x_1,1,x_2,y_2)$ is served, hence we have a death rate equal to $1$ for these job types. The creation of jobs is different now, as when a job is created, this creates a new type $(1,1,x_2,y_2)$ job at rate $\lambda q_t(x_2-1)q_t(y_2-1)$.

Jobs can also change type. Contrary to FCFS, a job type increases in all its entries when there is an arrival:  a type $(x_1,y_1,x_2,y_2)$ job becomes a type
$(x_1+1,y_1,x_2+1,y_2)$ or $(x_1,y_1+1,x_2,y_2+1)$ whenever an arrival occurs that places a replica at one of the servers of the tagged job. Each such event occurs at rate $2\lambda$.   
The changes due to a service completion in one of the two servers of the tagged job are the
same as in the FCFS case. The same holds for job type changes that occur when a buddy of 
another job present in the servers of the tagged job completes service.

\subsubsection{The set of ODEs}

Given the above discussion on the rates, we find that the set of ODEs given by \eqref{eq:ODEFCFS}
can be adapted to the following set of ODEs for $\gamma_t(x_1,y_1,x_2,y_2)$:
\begin{align}\label{eq:ODELCFS}
    \frac{d}{dt} \gamma_t(&x_1,y_1,x_2,y_2)= \lambda q_t(x_2-1) q_t(y_2-1) 
    1[x_1=1, y_1=1]- 4\lambda \gamma_t(x_1,y_1,x_2,y_2) \nonumber \\&+
    2\gamma_t(x_1-1,y_1,x_2-1,y_2) + 2\lambda \gamma_t(x_1,y_1-1,x_2,y_2-1) \nonumber\\
    &-\gamma_t(x_1,y_1,x_2,y_2) 
    \left[2+ \phi_t(x_2,x_2) - \frac{\gamma_t(x_1,1,x_2)}{\gamma_t(x_1,x_2)}
    + \phi_t(y_2,y_2) - \frac{\gamma_t(y_1,1,y_2)}{\gamma_t(y_1,y_2)}\right] 
    \nonumber \\ 
    &+\gamma_t(x_1+1,y_1,x_2+1,y_2) 
    \left[1+ \phi_t(x_1,x_2+1) \right] +\gamma_t(x_1,y_1+1,x_2,y_2+1) 
    \left[1+ \phi_t(y_1,y_2+1) \right]\nonumber \\
    &+\gamma_t(x_1,y_1,x_2+1,y_2) 
    \left[\phi_t(x_2+1,x_2+1) -\phi_t(x_1,x_2+1) \right] \nonumber \\
    &+\gamma_t(x_1,y_1,x_2,y_2+1) 
    \left[\phi_t(y_2+1,y_2+1) -\phi_t(y_1,y_2+1) \right],
\end{align}
where $\phi_t(x_a,x_b)=\sum_{x'=1}^{x_a} \gamma_t(x',1,x_b)/\gamma_t(x',x_b)$ has the same
form as $\kappa_t(x_a,x_b)$ for FCFS.
By summing this over $y_1,y_2,x_1$ and $x_2$ we obtain the same 
ODE  $\frac{d q_t}{dt} = 2\lambda - 2(1-q_0(t))$ for the mean queue length as in the FCFS model,
while \eqref{eq:indepq} can be recovered from \eqref{eq:ODELCFS} if we make a similar
independence assumption $\gamma_t(x_1,y_1,x_2,y_2)\approx \frac{2\gamma_t(x_1,x_2)\gamma_t(y_1,y_2)}{\bar q_t},$ as in the FCFS case.

\subsubsection{Model Validation}

The pair approximation for LCFS servers is validated using 
Tables \ref{tab:simulPAIR09LCFS}, \ref{tab:simulPAIR05LCFS} and \ref{tab:simulPAIR07LCFS}.
While the pair approximation clearly improves upon the mean field approximation, the
accuracy in case of LCFS servers is lower than in the setting with PS, FCFS or LPS($K$) servers.
This is somewhat expected as the correlation between the positions of two replicas of the same job is
strong in case of a LCFS server as both replicas start in position one. This is in contrast
to FCFS servers, where the positions are initially uncorrelated, or to PS servers, where the positions
are irrelevant.

\begin{table}[t]
    \centering
    \begin{tabular}{|r|c|c|c|c|c|c|c|c|}
    \hline 
        Number of jobs & mean & 1 & 2 & 3 & 5 & 10\\ \hline
         \hline
Simul $n=10^2$&   3.9944 & 1.2787e-01 &  1.3872e-01 &  1.3535e-01 &  1.0262e-01 &  2.0748e-02\\
   $n=10^3$& 3.9484 &  1.2803e-01&   1.3949e-01&   1.3658e-01&   1.0391e-01&   1.9867e-02\\
   $n=10^4$& 3.9415 &  1.2815e-01&   1.3963e-01&   1.3675e-01&   1.0404e-01&   1.9733e-02\\
   $n=10^5$& 3.9415 &  1.2815e-01&   1.3963e-01&   1.3675e-01&   1.0404e-01&   1.9733e-02\\\hline
  Pair Approx. & 3.6966 &  1.3135e-01&   1.4699e-01&   1.4635e-01&   1.0916e-01&   1.3939e-02\\\hline
  Mean Field & 3.3377 &  1.4177e-01&   1.6578e-01&   1.6496e-01&   1.0950e-01&   6.8438e-03\\
         \hline
    \end{tabular}
    \caption{Accuracy of Mean Field and Pair Approximation for LCFS servers, arrival rate $\lambda=0.9$. }
    \label{tab:simulPAIR09LCFS}
\end{table}

\section{Influence of the service disciplines}
\label{sec:numerical}

In this section, we provide numerical results to compare the performance of the various scheduling disciplines. \red{Unless specified otherwise, all reported numbers are obtained by computing the fixed points of the pair-approximations provided in Sections~\ref{sec:PS-approx} and \ref{sec:service-disciplines}. These fixed points are computed by numerically integrating the corresponding systems of ODEs.}

\subsection{Influence on the queue length distribution}

The previous sections show that our approximations provide accurate estimates of the queue length distributions and show that the service discipline has a true impact on the average queue length of the system. To push further the comparison between the various service disciplines, we show in Figure~\ref{fig:comparison_distrib} the queue length distribution for the various service disciplines.  These results show that when the load is not too high ($\lambda=0.7$), all policies behave 
similarly. This is because when $\lambda$ is small, many jobs have at least one of their replicas that is alone in a server where the service discipline has no impact. As $\lambda$ goes to $1$, most servers have a large queue length, and the service discipline has a high impact. This figure shows that the queue length distribution is minimal for FCFS, and maximal for LCFS, with PS and LPS(2) being in the middle.

\begin{figure}[t]
    \centering
    \begin{tabular}{ccc}
        \includegraphics[width=0.3\linewidth]{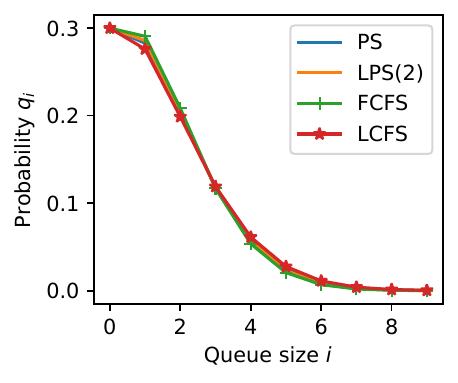}
        &\includegraphics[width=0.3\linewidth]{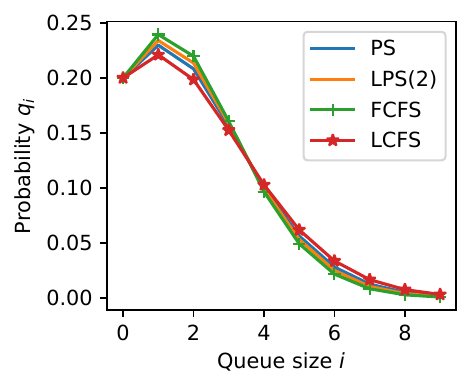}
        &\includegraphics[width=0.3\linewidth]{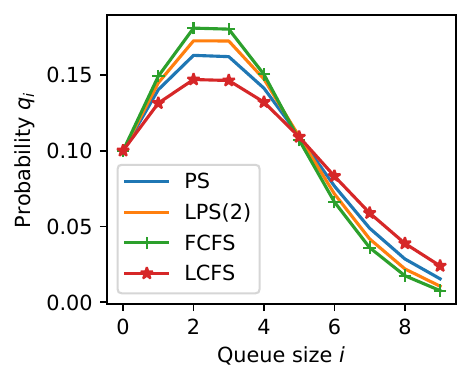}\\
        (a) $\lambda=0.7$
        &(b) $\lambda=0.8$
        &(c) $\lambda=0.9$
    \end{tabular}
    \caption{Queue length distributions of the various policies (computed by the pair approximations).}
    \label{fig:comparison_distrib}
\end{figure}

\subsection{Impact on the average queue length}

To study further this dependence on $\lambda$, we plot in Figure~\ref{fig:comparison_queuelength}, the average queue length as a function of $\lambda$. In Figure~\ref{fig:comparison_queuelength}(a), we compare the three policies PS, LCFS and FCFS. This figure confirms what is observed on Figure~\ref{fig:comparison_distrib}: FCFS always provides the smallest queue length and the difference is larger for larger loads. This panel does not show the numbers for LPS to improve readability. In Figure~\ref{fig:comparison_queuelength}(b), we quantify this difference by showing the ratio between the average queue length of each of the four policies (PS, LPS(2), LPS(5) and LCFS) and the one of FCFS. All numbers are higher than $1$ because FCFS has the smallest average queue length. As expected, LPS(2) has a behavior almost identical to FCFS (less than $5\%$ of difference in terms of average queue length, even for $\lambda=0.95$) whereas LCFS can lead to a queue length up to $25\%$ higher when $\lambda=0.95$. 

\red{To illustrate the accuracy of our approximations:
\begin{itemize}
    \item In Figure~\ref{fig:comparison_queuelength}(a), we also add simulation results for $\lambda\in\{0.3, 0.5, 0.7, 0.9\}$ and $n=10^6$ servers, as well as the approximation of the average queue length provided in \cite{gardner2017redundancy}. These numbers are -- qualitatively and quantitatively -- close to our pair approximation.
    \item In Figure~\ref{fig:comparison_queuelength}(b), the value for FCFS that we use as a denominator is our pair-approximation for FCFS. We also report a curve that is the ratio of the approximation for FCFS provided in \cite{gardner2017redundancy} (that is asymptotically exact) to our pair-approximation. It shows that the difference between the two approximations is smaller than $1\%$ for $\lambda\le0.95$.
\end{itemize}
Note that the approximation provided in \cite{gardner2017redundancy} is only for the average queue length, and only for FCFS. This explains why we cannot make a similar comparison when looking at queue length distribution or at other service disciplines. 
}

\begin{figure}[t]
    \centering
    \begin{tabular}{@{}cc@{}}
        \includegraphics[width=0.43\linewidth]{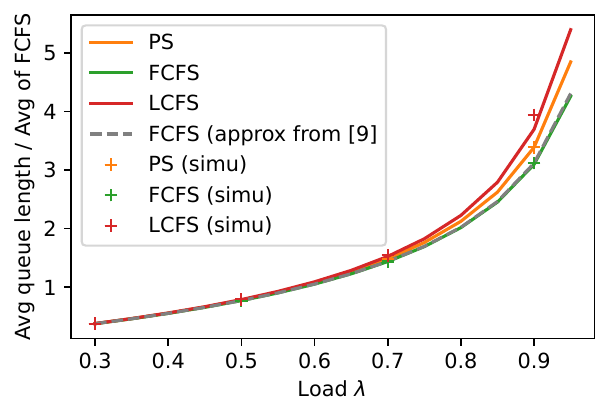}
        &\includegraphics[width=0.45\linewidth]{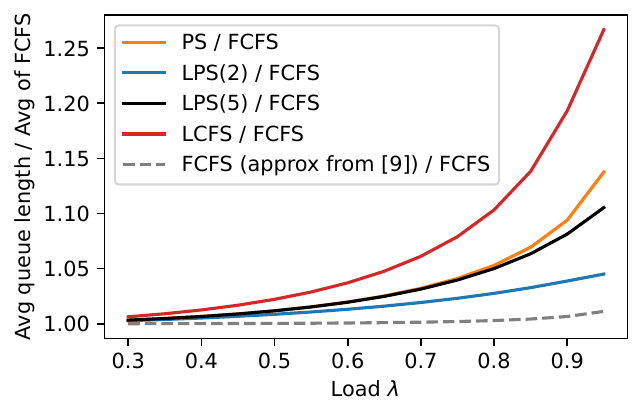}\\[-2pt]
        (a) Average queue length
        & (b) Ratio with respect to FCFS. 
    \end{tabular}
    \caption{Average queue length as a function of $\lambda$. \red{Unless specified otherwise, all numbers are computed by using the pair approximations that we derive in the paper.}.}
    \label{fig:comparison_queuelength}
\end{figure}

\subsection{Correlations between the buddies of a job}
\label{sec:correlations}

If the pair approximation provides accurate estimates of the queue length distribution, it can also be used to study the correlations between the two buddies of a jobs. This would not be possible by using a mean field approximation that precisely acts as if replicas were independent.  

In Figure~\ref{fig:correlations}, we study the correlations between the two buddies of a job for the three scheduling policies PS, FCFS and LCFS. For each policy, we plot the rate at which the buddy replica disappears as a function of the state of the first replica: For PS, the $x$-axis corresponds to the queue length of the first replica and the curve corresponds to $x\mapsto h(x)$. For FCFS or LCFS, it corresponds either to the position $x_1$ or to the queue length $x_2$, \emph{i.e.}, 
\begin{align*}
    \text{(Dashed) green curve: }x_1\mapsto \frac{\sum_{x_2,y_2} \delta(x_1, 1, x_2, y_2)}{\sum_{x, x_2,y_2} \delta(x_1, x, x_2, y_2)};
    &\quad \text{Orange curve: }
    x_2\mapsto\frac{\sum_{x_1,y_2} \delta(x_1, 1, x_2, y_2)}{\sum_{x, x_1,y_2} \delta(x_1, x, x_2, y_2)}.
\end{align*}
(or the same quantity with $\gamma$ for LCFS). 

We observe that for PS, the influence of the queue length is relatively limited (rates vary from $2.55$ to $2.85$), which explains why the performance of PS is close to the one of the independent model. FCFS is the only policy for which the rate at which the buddy disappears increases with the queue length of the first buddy. This could explain why FCFS helps reducing the queue length.

\begin{figure}[t]
    \begin{tabular}{ccc}
        \includegraphics[width=0.3\linewidth]{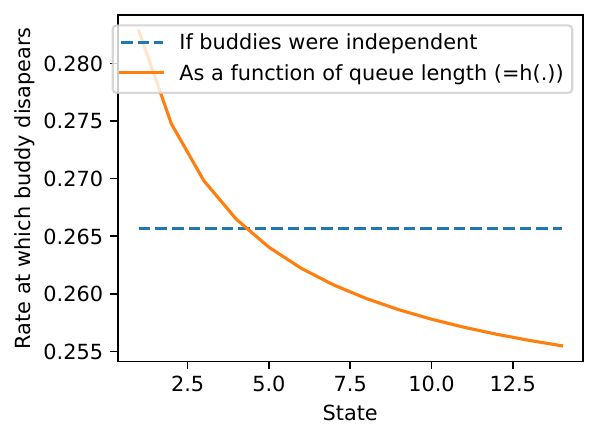}
        &\includegraphics[width=0.3\linewidth]{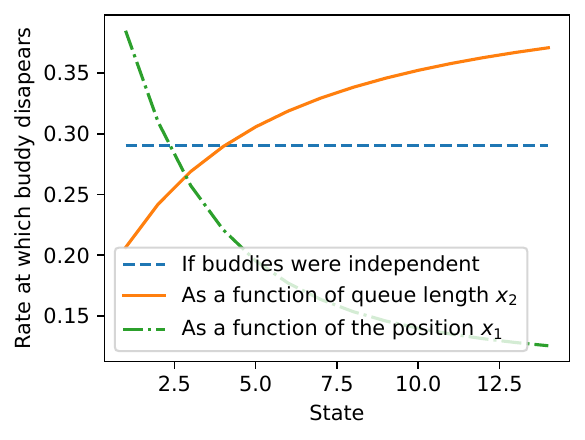}
        &\includegraphics[width=0.3\linewidth]{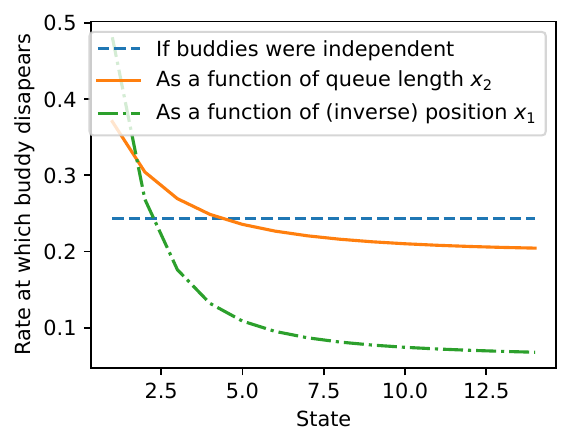}\\
        (a) PS& (b) FCFS & (c) LCFS
    \end{tabular}
    \caption{Rate at which the buddy replica disappears as a function of the state of the first job (as a function of the queue length for PS and of the position of the job or the queue length for LCFS or LCFS).}
    \label{fig:correlations}
\end{figure}

\section{Conclusions and future work}\label{sec:conclusions}
In this paper we presented the first approximation methods to study the impact of the service discipline on the queue length distribution in large queueing systems with redundancy. Although our models were shown to be highly accurate, they are not asymptotically exact. This work triggers two important lines of future work. First, from an application point of view, the next step is to use these models to study the impact of the token selection mechanism in Join-Idle-Queue load balancing with redundancy. From a more theoretical point of view, the dynamic random graph model introduced in Section~\ref{sec:PS-approx} can be viewed as a dynamic version of the classical Erdös-Renyi graph where we add arrivals and departures. Our work provides some insight on how to study the structural properties of this graph for large $n$. 

\begin{acks}
    This work is supported by the French National Research Agency (ANR) through REFINO Project under Grant ANR-19-CE23-0015 and by FWO (Fonds voor Wetenschappelijk Onderzoek).
\end{acks}

\bibliographystyle{acm}
\bibliography{thesis}


\appendix 

\section{Link with Join-Idle-Queue load balancing with redundancy}
\label{apx:JIQ-d}

While we believe that studying redundancy systems with different service disciplines is important on its own, we provide additional
motivation for the specific model with exponential job sizes and i.i.d.~replicas in this section. More specifically,
we argue that our redundancy models can be used to study the so-called Join-Idle-Queue (JIQ) load balancing policy with redundancy
as explained below.  In fact, this was the original problem that triggered us to look at these redundancy models.

Under the traditional JIQ policy \cite{lu1,stolyar1} each dispatcher maintains a list of servers ids that are considered to
be idle. When a job arrives at a dispatcher and the list is non-empty, the dispatcher assigns the incoming job to one
of the servers in the list and removes this server id from the list. Notice that several policies can be used to pick the server
id from the list. As will become clear further on, this policy will correspond to the service discipline in our redundancy queue,
for instance, we get a PS server if the id is selected at random.
Servers that become idle when the traditional JIQ is used, select a {\it single} dispatcher at random and their id is added to the list of the selected dispatcher. 

Clearly, under high loads there are very few idle servers and many dispatchers may have an idle
list. When a new job arrives at a dispatcher with an idle list, the dispatcher assigns the job at random. In  such case there
are two flavors of JIQ: {\it with and without withdrawals} \cite{wangJIQ, myMASCOTSpaper}. Without withdrawals an idle server will not remove its id from the
list of a dispatcher when it receives a job from another dispatcher. In this case it is shown in \cite{mitzenmacherJIQ}
that the policy used by the dispatcher to select a server id from its list impacts performance. When withdrawals are used,
servers remove their id from the list when receiving a job form another dispatcher. This means that any server id that
is on the list of a server necessarily corresponds to an idle server. In this case the policy used by the
dispatcher to select a server id from its list has no impact on the performance \cite{myMASCOTSpaper}.

When the load is high a possible approach to improve the performance of JIQ is to use a threshold such that when a
server has a queue length below the threshold, it adds its server id to the list of a randomly selected dispatcher \cite{mitzenmacherJIQ}.
Another possibility \cite{Harte} is that a server places its id at multiple, say $d$, randomly selected 
dispatchers when becoming idle. As soon as a server receives a job from any of these $d$ servers, it removes its id from all of the
other $d-1$ servers. This policy is called JIQ with redundancy, abbreviated as JIQ-Red($d$). It turns out that
under JIQ-Red($d$) the order in which the server selects a server id from its list again impacts performance.
It is however intuitively not clear which policy works best: random selection, FCFS, LCFS, etc? 

To see how this JIQ-Red($d$) problem is related to the model considered in this paper, think of the 
dispatcher as the queues, and the lists of server ids as the jobs in the queue.
Arrivals correspond to servers becoming idle. In such an event the server places its id at $d$ dispatchers. This corresponds
to the redundancy $d$ aspect of the queueing system. The job arrivals at the dispatchers correspond to service completions,
the manner in which the server id is selected to the service discipline and the removal by the server of its id from
the $d-1$ other servers corresponds to the cancellation-on-completion. As incoming jobs in the JIQ-Red($d$) are assumed
to correspond to a Poisson process, the corresponding service times are indeed exponential and different replicas have
independent sizes as each dispatcher is subject to its own independent Poisson process. 

There is however one difference between our queueing model and the JIQ-Red($d$) system. Under JIQ-Red($d$) a server can
also receive a job from a dispatcher that has an empty server id list. In this case the server will remove its id from
the $d$ dispatchers holding its id. Such an event can be incorporated in our model by adding exponential job abandonment.
In fact it is not hard to include abandonment in our redundancy models, but we decided to exclude them the keep the
presentation cleaner. We should also note that some of the rates in the JIQ-Red($d$) model are not known explicitly
(such as the abandonment rate), but this is something that can be dealt with by solving a sequence of models with fixed input parameters (as in our model).
In short, developing a means to analyze the impact of the service discipline in a redundancy model with exponential and 
independent replicas, should enable us to study the impact of the server selection id policy of the JIQ-Red($d$)
load balancing policy. 

We end this section by noting that  \cite{Harte} presents a mean field model for JIQ-Red($d$). However, this
model does not take the policy used by the dispatcher to select a server id into account. In fact if we think in terms
of the corresponding redundany model as discussed above, the model in \cite{Harte} corresponds to our
mean field model in Section \ref{sec:approx} if we add job abandonment.

\section{Pair approximations for $d>2$ replicas}\label{apx:d>2}
The pair approximations presented in this paper for PS, FCFS, LCFS and LPS with $d=2$ can
be generalized without much effort to a setting with $d>2$. The main downside of these generalizations
is that determining a fixed point becomes computationally expensive as the set of ODEs has
$d$ (for PS) or $2d$ (for FCFS, LCFS and LPS) dimensions.
For the triplet approximation such a generalization seems less obvious. In this section, we construct the generalization of the pair approximation with PS servers and provide a numerical evaluation of this approximation.

\subsection{Construction of the ODE}

When the servers use PS we make use of the variables $\pi_t^{(d)}(x_1,\ldots,x_d)$,
where the idea is that the $i$-th replica of a server is in a queue with length $x_i$.
These variables are symmetric as in the case with $d=2$.
We can then generalize \eqref{eq:pi(x)} to
\begin{align*}
    \pi_t^{(d)}(x_1)=\sum_{x_2,\ldots,x_d}\pi_t^{(d)}(x_1,\ldots,x_d),
\end{align*}
and \eqref{eq:condpi} to
\begin{align*}
    \pi_t^{(d)}(x_2|x_1)=\frac{1}{x_t^{(d)}(x_1)} \sum_{x_3,\ldots,x_d}\pi_t^{(d)}(x_1,x_2\ldots,x_d).
\end{align*}
The queue length can then be expressed by generalizing \eqref{eq:q_from_pi}:
\begin{align*}
    q_t^{(d)}(x)=\frac{d\pi_t^{(d)}(x)}{x}.
\end{align*}
Let $\vec x=(x_1,\ldots,x_d)$ and $e_i$ be equal to the $i$-th row of the order $d$ unity
matrix. Using these notations and the above quantities the set of ODEs given by \eqref{eq:ODEPS} becomes
\begin{align*}
      \frac{d \pi_t^{(d)}(\vec x)}{dt} &= \lambda \prod_{i=1}^d q_t^{(d)}(x_i-1) 
    +d\lambda \left[\sum_{i=1}^d \pi_t^{(d)}(\vec x - e_i)-d\pi_t^{(d)}(\vec x)\right] 
    \nonumber \\ 
    &\quad+ \sum_{i=1}^d \pi_t^{(d)}(\vec x + e_i)
    \left[ h_t^{(d)}(x_i+1)+ \frac{x_i}{x_i+1}\right]  
    - \pi_t(\vec x)\left[d+\sum_{i=1}^d  h_t(x_i)\right],
\end{align*}
with the convention that $\pi_t^{(d)}(\vec x)=0$ if one of the entries of 
$\vec x$ equals zero and where $h_t^{(d)}(x)$ is given by
\begin{align*}
     h_t^{(d)}(x) = (x-1)(d-1)\sum_{y'\geq 1} \pi_t^{(d)}(y'|x) \frac{1}{y'}.
\end{align*}

\subsection{Numerical evaluation}

\begin{table}[t]
    \centering
    \begin{tabular}{|r|c|c|c|c|c|c|c|c|}
    \hline 
        Number of jobs & mean & 1 & 2 & 3 & 5 & 10\\ \hline
         \hline
Simul $n=10^2$&          1.9374&   2.5592e-01  & 2.2734e-01 &  1.5648e-01 &  4.2942e-02 &  2.2749e-04\\
$n=10^3$&1.9072  & 2.5867e-01  & 2.3084e-01   &1.5766e-01  & 4.0735e-02  & 1.3987e-04\\
 $n=10^4$&   1.9046 &  2.5891e-01 &  2.3117e-01   &1.5780e-01&   4.0532e-02  & 1.3236e-04\\
  $n=10^5$&  1.9033  & 2.5905e-01 &  2.3124e-01   &1.5774e-01&   4.0449e-02  & 1.3136e-04\\ \hline
  Pair Approx. & 1.9033 &  2.5904e-01 &  2.3125e-01 &  1.5776e-01 &  4.0448e-02 &  1.3110e-04 \\ \hline
 Mean Field &  1.8947  & 2.6024e-01  & 2.3227e-01   &1.5776e-01  & 3.9746e-02 &  1.1981e-04\\
           \hline
    \end{tabular}
    \caption{Accuracy of Mean Field and Pair Approximation for PS servers, arrival rate $\lambda=0.8$ and $d=3$. }
    \label{tab:simulPAIR08PSd3}
\end{table}

Similarly to the case $d=2$, we used a numerical intergration of the ODE to compute the fixed point of the above ODE for the case $\lambda=0.9$ and $d=3$. We truncated the queue length to $30$ with gives an ODE of $30^3$ variables (instead of $30^2$ for $d=2$), which allows for a fast computation of the fixed point. In Table~\ref{tab:simulPAIR08PSd3}, we compare the values obtained by our pair approximation with values obtain by simulation and with the mean field approximation. This table illustrates that, as for $d=2$, the pair approximation for PS is very accurate.

\color{black}

\section{More on the Pair Approximation for PS servers}\label{apx:morePS}
\subsection{Analysis of $q(0)$}

From \eqref{eq:ODEPS}, we can derive the evolution of $\pi_t(x)$, by summing over $y\ge1$. This yields (for $x\ge1$)
\begin{align}\label{eq:sumy}
    \frac{d \pi_t(x)}{dt} &= \lambda q_t(x-1)  + 2\lambda(\pi_t(x-1)-\pi_t(x)) + \pi_t(x+1)\left[h_t(x+1)+\frac{x}{x+1}\right] \nonumber\\
    &\qquad+ \sum_{y\ge1} \pi_t(x,y+1)\frac{y}{y+1} - \pi_t(x)(2+h_t(x)),
\end{align}
where in the first equation we simplify the telescopic sum $\sum_{y\ge1}\pi_t(x,y+1)h_t(y+1) - \pi_t(x,y)h_t(y)=-\pi_t(x,1)h_t(1)=0$ because $h_t(1)=0$ by definition  \eqref{eq:h-pair}. 

If we further sum \eqref{eq:sumy} over $x\ge1$, the terms $\sum_{x\ge1}\pi_t(x+1)h(x+1)-\pi_t(x)h_t(x)=0$. This shows that
\begin{align}\label{eq:sumxy}
    \sum_{x\ge1} \frac{d \pi_t(x)}{dt} &= \lambda  + \sum_{x\ge1} \pi_t(x+1)  \frac{x}{x+1} + \sum_{y\ge1} \pi_t(y+1)\frac{y}{y+1} - 2\sum_{x\ge1} \pi_t(x) \\
    &= \lambda - 2 \sum_{x\ge1} \frac{\pi_t(x)}{x}\nonumber\\
    &= \lambda - (1-q_t(0)),
\end{align}
where the last line is due to \eqref{eq:q_from_pi}. 

Let $\bar q_t$ be the mean queue length at time $t$, then $\bar q_t = \sum_x x q_t(x) = 2\sum_{x}\pi_t(x)$ due to \eqref{eq:q_from_pi} and the above ODE~\eqref{eq:sumxy} can be stated as
\begin{align}\label{eq:ODEbarqt}
     \frac{d \bar q_t}{dt}  = 2\lambda - 2(1-q_t(0)),
\end{align}
Hence, for any fixed point we have $q(0)=2\sum_{x,y} \pi(x,y) = 1-\lambda$. This is consistent with the mean field approximation (and intuition) as the fraction of empty servers is expected to be $1-\lambda$.

\subsection{Recovering the mean field model from the pair approximation}

We noted earlier on that the pair approximation is more accurate than the mean field approximation.  Here, we show that the two approximation are closely linked. The ODE for \eqref{eq:sumy} is written for the quantities $\pi_t(x)=x q_t(x)/2$ while our original mean field approximation is written for the quantities $q_t(x)$. Multiplying the term \eqref{eq:sumy} by $x/2$, the pair approximation ODE can be rewritten for $q_t$ as follows:
\begin{align*}
    \frac{d q_t(x)}{dt} &= 2\lambda (q_t(x-1)-q_t(x)) + \frac{q_t(x+1) (x+1)}{x}\left[h_t(x+1)+\frac{x}{x+1}\right] \nonumber\\
    &\qquad+ \sum_{y\ge1} \frac2x \pi_t(x,y+1)\frac{y}{y+1} - q_t(x)(2+h_t(x)),\\
    &= 2\lambda (q_t(x-1)-q_t(x)) + q_t(x+1)-q_t(x) + q_t(x+1)(x+1)g_t(x) - q_t(x)x g_t(x),
\end{align*}
where $g_t(x)=\sum_{y' \geq 1} \pi_t(y'|x) \frac{1}{y'}$, which is equal to $h_t(x)/(x-1)$ except for $x=1$.

This equation is almost the same as the mean field equation \eqref{eq:indepq}, except for the last two terms.
Here, we show that we recover the mean field approximation if we make the approximation that $\pi_t(y|x)$ does not depend on $x$, that is,
if $\pi_t(y|x) \approx \pi_t(y)/\sum_z \pi_t(z)$. 
Indeed, under such an approximation, we obtain that 
\begin{align*}
    g_t(x) = \sum_{y' \geq 1} \pi_t(y'|x)\frac{1}{y'} \approx \sum_{y'\ge1} \frac{2\pi_t(y')}{y'\bar{q}_t} = \frac{1}{\bar{q}_t}\sum_{y'\ge1} q_t(y')=\frac{1-q_t(0)}{\bar{q}_t}.
\end{align*}

\section{Recovering the pair approximation from the triplet approximation}
\label{apx:pair from triplet}

Here, we show that one can recover the pair approximation from the triplet approximation when assuming an extra independence assumption (Equation~\eqref{eq:PA} below).

To show that, we introduce two new quantities $k(y,x)=\sum_v c(v|y,x) \frac{x-1}{v}=\frac{x-1}{c(y,x)}\sum_{v} \frac{c(y,x,v)}{v}$ and $\ell(x,y,z)=\sum_v c(v|x,y,z) \frac{y-2}{v}$. By using these notations, the triplet ODE~\eqref{eq:ODE_triplet} can be writen as
\begin{align*} 
    \frac{d}{dt} c(x,y,z) = ~&\lambda q(x-1)q(y-1)c(z|y-1)(y-1) + \lambda q(y-1)q(z-1)c(x|y-1) (y-1) \nonumber\\
    & -6\lambda c(x,y,z) + 2\lambda c(x-1,y,z) + 2\lambda c(x,y-1,z) +2\lambda c(x,y,z-1)\nonumber \\
    &-3 c(x,y,z) + c(x+1,y,z)\frac{x}{x+1}+ c(x,y+1,z)\frac{y-1}{y+1} + c(x,y,z+1)\frac{z}{z+1} \nonumber\\
    &-c(x,y,z) \left(k(y,x)  + k(y,z) + \ell(x,y,z) \right)\nonumber \\
    &+c(x+1,y,z) k(y, x+1) +c(x,y,z+1) k(y, z+1) +c(x,y+1,z) \ell(x,y+1,z),
\end{align*}
by just replacing the corresponding quantities by $k$ and $\ell$. 

Recall that $\pi(x,y)=c(x,y)/(y-1) = \sum_z c(x,y,z)/(y-1)$ . Hence, to obtain an ODE for $\pi$, we sum the above equation over $z$, which shows that (for $y\ge2$):
\begin{align*}
    \frac{d}{dt} c(x,y) = ~&\lambda q(x-1)q(y-1)(y-1) + \lambda q(y-1)c(x|y-1) (y-1) \nonumber\\
    & -6\lambda c(x,y) + 2\lambda c(x-1,y) + 2\lambda c(x,y-1) + 2\lambda c(x,y) \nonumber \\
    &-3 c(x,y) + c(x+1,y)\frac{x}{x+1}+  c(x,y+1)\frac{y-1}{y+1} + \sum_z c(x,y,z+1)\frac{z}{z+1} \nonumber\\
    &-c(x,y) k(y,x)  -\sum_z c(x,y,z) k(y,z) -\sum_z c(x,y,z) \ell(x,y,z) \nonumber \\
    &+c(x+1,y) k(y, x+1) + \sum_z c(x,y,z+1) k(y, z+1) + \sum_z c(x,y+1,z) \ell(x,y+1,z)\\
\end{align*}
There are a few simplification to the above equation: 
\begin{itemize}
    \item The last term of the second line removes $2\lambda$ of the $6\lambda$ of the first term.
    \item On the third line, the third term is equal to $c(x,y+1)\frac{y-1}{y+1}=\frac{y-1}{y}c(x,y+1)\frac{y}{y+1}=(y-1)\pi(x,y+1)\frac{y}{y+1}$ and the last term is equal to $\sum_z c(x,y,z+1)\frac{z}{z+1}=c(x,y) - \sum_{z\ge1}c(x,y,z)\frac{1}{z}$ because the sum can starts at $z=0$ (as $c(.,.,0)=0$). 
    \item The last term of the fourth cancels with the second term of the fifth line because, as before, the sum can starts at $z=0$ or $z=1$ because as  $c(.,.,0)=0$.
\end{itemize}

By using that $\pi(x,y)=c(x,y)/(y-1)$, this implies that
\begin{align}
    \label{eq:triplet_pi(x,y)}
    \frac{d}{dt} \pi(x,y) = ~&\lambda q(x-1)q(y-1) + \lambda q(y-1)c(x|y-1) - \underbrace{\frac{2\lambda}{y-1}\pi(x,y-1)}_{\text{(*) moved from second line. Cancels with previous term.}}\nonumber\\
    & -4\lambda \pi(x,y) + 2\lambda \pi(x-1,y) + 2\lambda \pi(x,y-1)  \nonumber \\
    &-2 \pi(x,y) + \pi(x+1,y)\frac{x}{x+1}+  \pi(x,y+1)\frac{y}{y+1} - \frac{1}{y-1} \sum_z c(x,y,z)\frac{1}{z} \nonumber\\
    &- \frac{1}{y-1} c(x,y) k(y,x)   - \frac{1}{y-1} \sum_z c(x,y,z) \ell(x,y,z) \nonumber \\
    &+\frac{1}{y-1} c(x+1,y) k(y, x+1) + \frac{1}{y-1} \sum_z c(x,y+1,z) \ell(x,y+1,z)
\end{align}
where the term (*) comes from the third term of the second line that can be rewritten as $2\lambda$ times $c(x,y-1)=(y-2)\pi(x,y-1)=(y-1)\pi(x,y-1)-\pi(x,y-1)$.

By reorganizing the terms, the above equation rewrites as:
\begin{align*}
    \frac{d}{dt} \pi(x,y) =~&
    \lambda q(x-1)q(y-1) + 2\lambda [\pi(x-1,y) +  \pi(x,y-1)-2 \pi(x,y)]  \nonumber \\
    &+ \pi(x+1,y)[\frac{x}{x+1} + k(y, x+1)] + \pi(x,y+1)\frac{y}{y+1} \\
    &-\pi(x,y) [2+k(y,x) ]\nonumber\\
    &+ \frac{1}{y-1} \sum_z c(x,y+1,z) \ell(x,y+1,z)- \frac{1}{y-1} \sum_z c(x,y,z)\frac{1}{z}  - \frac{1}{y-1} \sum_z c(x,y,z) \ell(x,y,z) 
\end{align*} 
Let $m(x,y)=\frac{1}{y-2} \sum_z c(x,y,z) \ell(x,y,z)/\pi(x,y)$ for $y\ge2$ and $m(x,2)=0$. The last line of the above equation is equal to \begin{align*}
    &\frac1{y-1}\pi(x,y)\sum_z \frac{c(x,y,z)}{z \pi(x,y)} - \frac{y-2}{y-1}\frac{1}{y-2}\sum_z c(x,y,z) \ell(x,y,z)\\
    &= -\frac1{y-1}\pi(x,y)k(x,y) - \pi(x,y)m(x,y) + \frac1{y-1}\pi(x,y)m(x,y).
\end{align*}
Hence, using that $m(x,y+1)=\frac{1}{y-1} \sum_z c(x,y+1,z) \ell(x,y+1,z)/\pi(x,y+1)$, we can rewrite the total sum as
\begin{align}
    \label{eq:pair_from_triplet}
    \frac{d}{dt} \pi(x,y) =~&
    \lambda q(x-1)q(y-1) + 2\lambda [\pi(x-1,y) + \pi(x,y-1)-2 \pi(x,y)]  \nonumber \\
    &+ \pi(x+1,y)\left[k(y, x+1) + \frac{x}{x+1}\right] + \pi(x,y+1)\left[m(x,y+1) + \frac{y}{y+1}\right] \nonumber\\
    &-\pi(x,y) [2+k(y,x) +m(x,y)]\nonumber\\
    &+ \frac{\pi(x,y)}{y-1}\left( m(x,y) - k(x,y)\right),
\end{align} 
The ODE~\eqref{eq:pair_from_triplet} is close to the pair approximation equation but with two modifications:
\begin{itemize}
    \item The quantity $h(x)$ of the pair approximation is replaced by $k(y,x)$ and the quantity $h(y)$ is replaced by $m(x,y+1)$.
    \item the last line does not exist for the pair approximation.
\end{itemize}
The pair approximation consists in assuming that when three servers are connected, their state only depends on the middle node (note that this only make sense when $y\ge2$).  Written in terms of $c_t(x,y,z)$, this means that for $y\ge2$ we would have:
\begin{align}
    \label{eq:PA}
    c_t(x,y,z) \approx c_t(x,y) c_t(z | y) = (y-1) \pi_t(x,y) \frac{\pi_t(y,z)}{\pi_t(y)},
\end{align}
where by construction $c_t(z|y) = \pi_t(z|y)= \pi_t(y,z)/\pi_t(y)$, see \eqref{eq:condpi}.

A straightforward computation shows that if one assumes \eqref{eq:PA}, then one has $m(x,y)=k(x,y)=h(y)$. In such a case, the ODE~\eqref{eq:pair_from_triplet} corresponds exactly to the triplet approximation.

\section{More on the Pair Approximation for FCFS servers}\label{apx:moreFCFS}
\subsection{Analysis of $q(0)$}

Summing \eqref{eq:ODEFCFS} $y_1$ and $y_2$ and simplifying yields
\begin{align}\label{eq:sumyFCFS}
    \frac{d \delta_t(x_1,x_2)}{dt} &= \lambda q_t(x_2-1) 1[x_1=x_2]
    - 2\lambda \delta_t(x_1,x_2) +
    2\lambda\delta_t(x_1,x_2-1) \nonumber\\
    &-\delta_t(x_1,x_2) + \delta_t(x_1+1,x_2+1)-\delta_t(x_1,1,x_2) \nonumber \\ 
    &-\delta_t(x_1,x_2) \left[\kappa_t(x_2,x_2) - \frac{\delta_t(x_1,1,x_2)}{\delta_t(x_1,x_2)}
    \right]
    \nonumber \\ 
    &+\delta_t(x_1+1,x_2+1) \kappa_t(x_1,x_2+1)  
    +\delta_t(x_1,x_2+1) 
    \left[\kappa_t(x_2+1,x_2+1) -\kappa_t(x_1,x_2+1) \right].
    \end{align}

If we further sum \eqref{eq:sumyFCFS} over $x_1$ and $x_2$, we find
\begin{align}\label{eq:sumxyFCFS}
    \sum_{x_1,x_2} \frac{d \delta_t(x_1,x_2)}{dt} &= 
    \lambda - 2\delta_t^{(pos)}(1,x_2)= \lambda -(1-q_t(0)),
\end{align}
where the second equality is found by summing  \eqref{eq:qxFCFS} over $x$. 

Let $\bar q_t$ be the mean queue length at time $t$, then $\bar q_t = \sum_{x} x q_t(x) = 2\sum_{x_1,x}\delta_t(x_1,x)$ due to \eqref{eq:qxFCFS2}
and the above ODE can be stated as 
\begin{align}\label{eq:ODEbarqtFCFS}
     \frac{d \bar q_t}{dt}  = 2\lambda - 2(1-q_t(0)),
\end{align}
similar to the PS case.
Hence, for any fixed point we have $q(0)= 1-\lambda$.

\subsection{Recovering the mean field model from the pair approximation}
The independence assumption that we used in Section~\ref{sec:approx} to construct the mean field approximation consists in saying that the position of two replicas are independent, that is
\[ \delta_t(x_1,y_1,x_2,y_2)\approx\frac{\delta_t(x_1,x_2)\delta_t(y_1,y_2)}{\sum_{z_1,z_2}\delta_t(z_1,z_2)} = \frac{2\delta_t(x_1,x_2)\delta_t(y_1,y_2)}{\bar q_t}, \]
where $\bar q_t = \sum_{x} x q_t(x) = 2\sum_{x_1,x_2}\delta_t(x_1,x_2)$ due to \eqref{eq:qxFCFS2} is the mean queue length. The independence assumption also implies that
$\kappa_t(x_a,x_b) \approx 2 x_a \delta_t^{(pos)}(1)/\bar q_t$. Plugging this into \eqref{eq:sumyFCFS} and  summing over $x_1$ yields
\begin{align}
    \frac{d \delta_t^{(ql)}(x_2)}{dt}& \approx
    \lambda q_t(x_2-1) -2 \lambda (\delta_t^{(ql)}(x_2)-\delta_t^{(ql)}(x_2-1))
    - \delta_t^{(ql)}(x_2)  \nonumber \\
    &+\sum_{x_1} \delta_t(x_1+1,x_2+1) - \delta_t^{(ql)}(x_2) \frac{2\delta_t^{(pos)}(1)}{\bar q_t}
     - \delta_t^{(ql)}(x_2) (x_2-1) \frac{2\delta_t^{(pos)}(1)}{\bar q_t}
    \nonumber \\
    &+\sum_{x_1} \delta_t(x_1+1,x_2+1)x_1  \frac{2\delta_t^{(pos)}(1)}{\bar q_t}
    +\sum_{x_1} \delta_t(x_1,x_2+1)(x_2-x_1+1)  \frac{2\delta_t^{(pos)}(1)}{\bar q_t} \nonumber \\
    &= \lambda q_t(x_2-1) -2 \lambda (\delta_t^{(ql)}(x_2)-\delta_t^{(ql)}(x_2-1))
    - \delta_t^{(ql)}(x_2)  \nonumber \\
    &+ \delta_t^{(ql)}(x_2+1)-\delta_t(1,x_2+1) -
    \frac{2\delta_t^{(pos)}(1)}{\bar q_t} \left( \delta_t^{(ql)}(x_2) x_2 -
    \delta_t^{(ql)}(x_2+1) x_2 \right).
\end{align}
Using \eqref{eq:qxFCFS2} we can transform this into a set of ODEs for $q_t(x)$ for $x > 0$
by noting that $q_t(x)=2\delta_t(1,x)$ and $2\delta_t^{(pos)}(1)=1-q_t(0)$:
\begin{align}\label{eq:indepqFCFS}
    \frac{d q_t(x)}{dt} &= q_t(x+1) - q_t(x) +2\lambda \left[q_t(x-1)-q_t(x)\right] 
    +  (1-q_t(0)) \frac{((x+1)q_t(x+1)-xq_t(x)) }{\bar q_t},
\end{align}
while $q_t(0)=1-\sum_{x\geq 1} q_t(x)$, which is the same as the mean field approximation \eqref{eq:indepq}.

This shows that, as for PS, the mean field ODE can be recovered from the pair approximation for FCFS if we assume independence. Note, however, that solving numerically \eqref{eq:indepq}, \eqref{eq:ODEPS} and \eqref{eq:ODEFCFS} do not lead to the same fixed point. This implies that without the independent assumption, the PS and FCFS models do not yield the same performance.

\section{More Model Validation Numbers}
\label{apx:more-numerical}

This section presents additional numerical results to assess the accuracy of the various approximations as a function of the number of servers. The tables presented here are the same as the ones presented in the main body of the paper but with other parameters values.  In more details, the tables presented here are:
\begin{itemize}
    \item Table~\ref{tab:simulPAIR05} and Table~\ref{tab:simulPAIR07} complement Table~\ref{tab:simulPAIR09} and validate the accuracy of the pair and triplet approximations for $\lambda=0.5$ and $\lambda=0.7$.
    \item Table~\ref{tab:simulPAIR05FCFS} and Table~\ref{tab:simulPAIR07FCFS} complement Table~\ref{tab:simulPAIR09FCFS} and validate the accuracy of the pair approximation for FCFS for $\lambda=0.5$ and $\lambda=0.7$.
    \item Table~\ref{tab:simulPAIR05LCFS} and Table~\ref{tab:simulPAIR07LCFS} do the same for LCFS.
    \item Table~\ref{tab:simulPAIR05LPS_K2} and Table~\ref{tab:simulPAIR07LPS_K2} do the same for LPS(2).
\end{itemize}
Finally, Figure~\ref{fig:non-asympto-exact2} illustrates how our various approximation can serve to estimate the average queue length for FCFS, LCFS and LPS. It complements Figure~\ref{fig:non-asympto-exact} that focuses only on PS. This figure shows that the accuracy of the pair approximation for LPS and FCFS are very good even if not as precise as the one for PS. The accuracy of the pair approximation for LCFS is somehow less good because of the strong correlations between replicas. \red{For FCFS, we also plot the numbers from \cite{gardner2017redundancy} that are shown to be asymptotically exact in \cite{shneer2020large}. This graph shows that our approximation for FCFS is very close to the one of \cite{gardner2017redundancy} while being different. }

\begin{table}[ht]
    \centering
    \begin{tabular}{|r|c|c|c|c|c|c|c|c|}
    \hline 
        Number of jobs & mean & 1 & 2 & 3 & 5 & 10\\ \hline
\hline
Simul $n=10^2$&7.8368e-01  & 3.0292e-01 &  1.3327e-01 &  4.6290e-02  & 3.2962e-03 &  6.1147e-07\\
 $n=10^3$&   7.8181e-01 &  3.0353e-01  & 1.3336e-01&   4.6040e-02  & 3.1801e-03&   4.3798e-07\\
 $n=10^4$&   7.8145e-01 &  3.0358e-01  & 1.3335e-01&   4.5998e-02  & 3.1620e-03&   4.2604e-07\\
 $n=10^5$&   7.8145e-01 &  3.0358e-01  & 1.3335e-01 &  4.5998e-02  & 3.1620e-03&   4.2604e-07\\ \hline
Triplet App. &7.8149e-01  & 3.0358e-01 &  1.3336e-01  & 4.6003e-02  & 3.1618e-03  & 4.3615e-07 \\\hline
Pair Approx.&   7.8142e-01 &  3.0361e-01  & 1.3336e-01 &  4.5992e-02  & 3.1590e-03&   4.3449e-07\\\hline
Mean Field&   7.7891e-01 &  3.0452e-01  & 1.3334e-01 &  4.5573e-02  & 3.0344e-03&   3.6906e-07\\         \hline
    \end{tabular}
    \caption{Accuracy of Mean Field, Pair and Triplet Approximation for PS servers, arrival rate $\lambda=0.5$. }
    \label{tab:simulPAIR05}
\end{table}

\begin{table}[ht]
    \centering
    \begin{tabular}{|r|c|c|c|c|c|c|c|c|}
    \hline 
        Number of jobs & mean & 1 & 2 & 3 & 5 & 10\\ \hline
        \hline
Simul $n=10^2$ &1.4924  & 2.8191e-01  & 2.0210e-01  & 1.1798e-01  & 2.5147e-02  & 8.2669e-05\\
 $n=10^3$ &   1.4851 & 2.8275e-01 &  2.0301e-01 &  1.1816e-01  & 2.4572e-02 &  6.5377e-05\\
 $n=10^4$ &   1.4845 &  2.8283e-01 &  2.0310e-01&   1.1819e-01 &  2.4514e-02 &  6.3963e-05\\
 $n=10^5$ &   1.4845 &  2.8283e-01 &  2.0310e-01&   1.1819e-01 &  2.4514e-02 &  6.3963e-05\\\hline
Triplet App. & 1.4845 &  2.8279e-01  & 2.0312e-01  & 1.1821e-01   &2.4512e-02   &6.3623e-05\\\hline
Pair Approx. &   1.4841 &  2.8288e-01 &  2.0315e-01&   1.1820e-01 &  2.4486e-02 &  6.3272e-05\\ \hline
Mean Field &   1.4730 &  2.8471e-01 &  2.0436e-01&   1.1795e-01 &  2.3606e-02 &  5.2281e-05 \\ 
         \hline
    \end{tabular}
    \caption{Accuracy of Mean Field, Pair and Triplet Approximation for PS servers, arrival rate $\lambda=0.7$. }
    \label{tab:simulPAIR07}
\end{table}

\begin{table}[ht]
    \centering
    \begin{tabular}{|r|c|c|c|c|c|c|c|c|}
    \hline 
        Number of jobs & mean & 1 & 2 & 3 & 5 & 10\\ \hline
         \hline
Simul $n=10^2$ &       7.7493e-01   &3.0609e-01  & 1.3318e-01&   4.4843e-02 &  2.8666e-03  & 2.9192e-07\\
$n=10^3$ &  7.7275e-01  & 3.0677e-01   &1.3324e-01&   4.4506e-02  & 2.7435e-03  & 2.7577e-07\\
$n=10^4$ &  7.7258e-01  & 3.0684e-01   &1.3324e-01&   4.4492e-02  & 2.7321e-03  & 2.3708e-07\\
$n=10^5$ &  7.7258e-01  & 3.0684e-01   &1.3324e-01&   4.4492e-02  & 2.7321e-03 &  2.3708e-07\\\hline
\red{2$\lambda E[T]$ \cite{gardner2017redundancy}} & \red{7.7259e-01}  &  -& - &-   &-  & -  \\\hline
Pair Approx. & 7.7233e-01 &  3.0694e-01 &  1.3326e-01 &  4.4449e-02  & 2.7170e-03  & 2.3489e-07\\\hline
Mean Field &  7.7891e-01  & 3.0452e-01   &1.3334e-01&   4.5573e-02  & 3.0344e-03 &  3.6906e-07\\
 \hline
    \end{tabular}
    \caption{Accuracy of Mean Field and Pair Approximation for FCFS servers, arrival rate $\lambda=0.5$. }
    \label{tab:simulPAIR05FCFS}
\end{table}

\begin{table}[ht]
    \centering
    \begin{tabular}{|r|c|c|c|c|c|c|c|c|}
    \hline
          Number of jobs & mean & 1 & 2 & 3 & 5 & 10\\ \hline
         \hline
Simul $n=10^2$  &  1.4491 &  2.8916e-01&   2.0698e-01&   1.1698e-01 &  2.1677e-02 &  3.6800e-05\\
 $n=10^3$ & 1.4409 &  2.9029e-01 &  2.0796e-01 &  1.1696e-01  & 2.0995e-02  & 2.9014e-05\\
  $n=10^4$& 1.4400 &  2.9040e-01 &  2.0805e-01 &  1.1696e-01  & 2.0923e-02  & 2.8286e-05\\
  $n=10^5$& 1.4400 &  2.9040e-01 &  2.0805e-01 &  1.1696e-01  & 2.0923e-02  & 2.8286e-05\\\hline
\red{2$\lambda E[T]$ \cite{gardner2017redundancy}} & \red{1.4399}  &  -& - &-   &-  & -  \\\hline
 Pair Approx. & 1.4379 &  2.9068e-01 &  2.0834e-01 &  1.1697e-01  & 2.0736e-02  & 2.6538e-05\\\hline
Mean Field  & 1.4730 &  2.8471e-01 &  2.0436e-01 &  1.1795e-01  & 2.3606e-02  & 5.2281e-05\\
    \hline 
    \end{tabular}
    \caption{Accuracy of Mean Field and Pair Approximation for FCFS servers, arrival rate $\lambda=0.7$. }
    \label{tab:simulPAIR07FCFS}
\end{table}

\begin{table}[ht]
    \centering
    \begin{tabular}{|r|c|c|c|c|c|c|c|c|}
    \hline 
        Number of jobs & mean & 1 & 2 & 3 & 5 & 10\\ \hline
         \hline
 Simul $n=10^2$&  7.9424e-01&   2.9940e-01 &  1.3326e-01&   4.7894e-02&   3.8482e-03&   1.0922e-06\\
    $n=10^3$&7.9185e-01&   2.9993e-01 &  1.3330e-01&   4.7638e-02&   3.7131e-03&   8.2594e-07\\
    $n=10^4$&7.9182e-01&   2.9998e-01 &  1.3334e-01&   4.7626e-02&   3.7045e-03&   8.2407e-07\\
   $n=10^5$& 7.9182e-01&   2.9998e-01 &  1.3334e-01&   4.7626e-02&   3.7045e-03&   8.2407e-07\\\hline
  Pair Approx. & 7.8945e-01&   3.0072e-01 &  1.3343e-01&   4.7300e-02&   3.5621e-03&   6.9569e-07\\\hline
  Mean Field & 7.7891e-01&   3.0452e-01 &  1.3334e-01&   4.5573e-02&   3.0344e-03&   3.6906e-07\\
 \hline
    \end{tabular}
    \caption{Accuracy of Mean Field and Pair Approximation for LCFS servers, arrival rate $\lambda=0.5$. }
    \label{tab:simulPAIR05LCFS}
\end{table}

\begin{table}[ht]
    \centering
    \begin{tabular}{|r|c|c|c|c|c|c|c|c|}
    \hline
          Number of jobs & mean & 1 & 2 & 3 & 5 & 10\\ \hline
         \hline
Simul $n=10^2$&   1.5536&   2.7299e-01&   1.9575e-01 &  1.1849e-01 &  2.9617e-02  & 2.0176e-04\\
  $n=10^3$& 1.5453&   2.7375e-01&   1.9656e-01 &  1.1870e-01 &  2.9133e-02  & 1.6615e-04\\
  $n=10^4$& 1.5450&   2.7379e-01&   1.9664e-01 &  1.1874e-01 &  2.9099e-02  & 1.6391e-04\\
  $n=10^5$& 1.5450&   2.7379e-01&   1.9664e-01 &  1.1874e-01 &  2.9099e-02  & 1.6391e-04\\\hline
  Pair Approx. & 1.5259&   2.7606e-01&   1.9877e-01 &  1.1904e-01 &  2.7697e-02  & 1.1591e-04\\\hline
  Mean Field &  1.4730&   2.8471e-01&   2.0436e-01 &  1.1795e-01 &  2.3606e-02  & 5.2281e-05\\
    \hline 
    \end{tabular}
    \caption{Accuracy of Mean Field and Pair Approximation for LCFS servers, arrival rate $\lambda=0.7$. }
    \label{tab:simulPAIR07LCFS}
\end{table}

\begin{table}[ht]
    \centering
    \begin{tabular}{|r|c|c|c|c|c|c|c|c|}
    \hline 
        Number of jobs & mean & 1 & 2 & 3 & 5 & 10\\ \hline
         \hline
Simul $n=10^2$&         7.8140e-01 &  3.0386e-01  & 1.3328e-01 &  4.5913e-02   &3.1709e-03  & 4.6516e-07\\
  $n=10^3$ & 7.7910e-01&   3.0444e-01 &  1.3329e-01&   4.5606e-02 &  3.0502e-03&   4.1686e-07\\
  $n=10^4$ & 7.7900e-01&   3.0449e-01 &  1.3336e-01&   4.5592e-02 &  3.0376e-03&   3.6451e-07\\
  $n=10^5$ & 7.7900e-01&   3.0449e-01 &  1.3336e-01&   4.5592e-02 &  3.0376e-03&   3.6451e-07\\\hline
 Pair Approx.  &7.7895e-01&   3.0450e-01 &  1.3334e-01&   4.5582e-02 &  3.0352e-03&   3.6858e-07\\\hline
 Mean Field & 7.7891e-01&   3.0452e-01 &  1.3334e-01&   4.5573e-02 &  3.0344e-03&   3.6906e-07\\
         \hline
    \end{tabular}
    \caption{Accuracy of Mean Field and Pair Approximation for LPS(2) servers, arrival rate $\lambda=0.5$. }
    \label{tab:simulPAIR05LPS_K2}
\end{table}

\begin{table}[ht]
    \centering
    \begin{tabular}{|r|c|c|c|c|c|c|c|c|}
    \hline 
        Number of jobs & mean & 1 & 2 & 3 & 5 & 10\\ \hline
         \hline
Simul $n=10^2$ &1.4758e+00  & 2.8460e-01  & 2.0417e-01  & 1.1788e-01  & 2.3797e-02 &  5.6935e-05\\
  $n=10^3$ & 1.4664e+00 &  2.8575e-01&   2.0512e-01&   1.1784e-01&   2.3067e-02  & 4.5098e-05\\
 $n=10^4$ &  1.4654e+00 &  2.8587e-01&   2.0521e-01&   1.1781e-01&   2.2992e-02  & 4.5006e-05\\
  $n=10^5$ & 1.4654e+00 &  2.8587e-01&   2.0521e-01&   1.1781e-01&   2.2992e-02  & 4.5006e-05\\\hline
 Pair Approx. &  1.4656e+00 &  2.8585e-01&   2.0523e-01&   1.1785e-01&   2.2996e-02  & 4.4963e-05\\\hline
  Mean Field &  1.4730e+00 &  2.8471e-01&   2.0436e-01&   1.1795e-01&   2.3606e-02  & 5.2281e-05\\
         \hline
    \end{tabular}
    \caption{Accuracy of Mean Field and Pair Approximation for LPS(2) servers, arrival rate $\lambda=0.7$. }
    \label{tab:simulPAIR07LPS_K2}
\end{table}

\begin{figure}[ht]
    \centering
    \begin{tabular}{@{}c@{}c@{}c@{}}
        \includegraphics[width=0.32\linewidth]{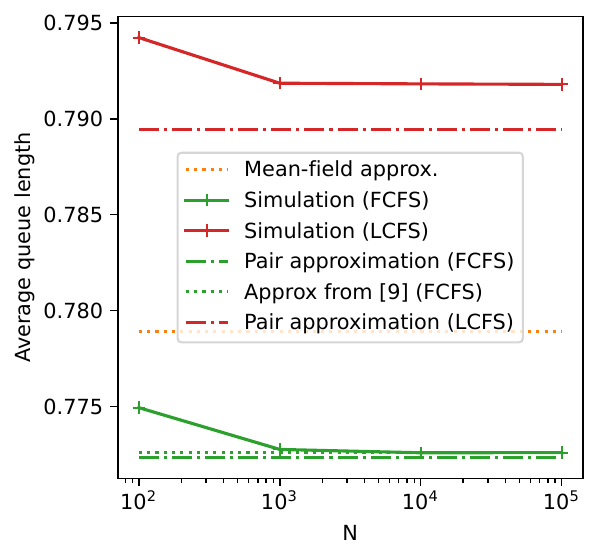}
        &\includegraphics[width=0.32\linewidth]{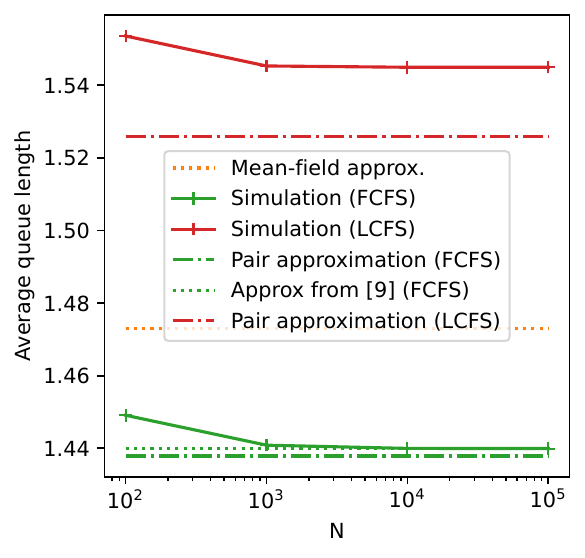}
        &\includegraphics[width=0.32\linewidth]{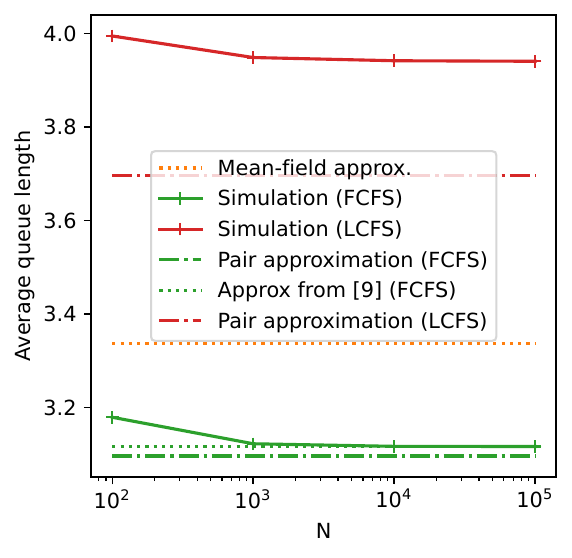}
        \\
        (a)  FCFS and LCFS, $\lambda=0.5$
        & (b) FCFS and LCFS, $\lambda=0.7$
        & (c) FCFS and LCFS, $\lambda=0.9$\\~\\
        \begin{tabular}{@{}c@{}}
            \includegraphics[width=0.32\linewidth]{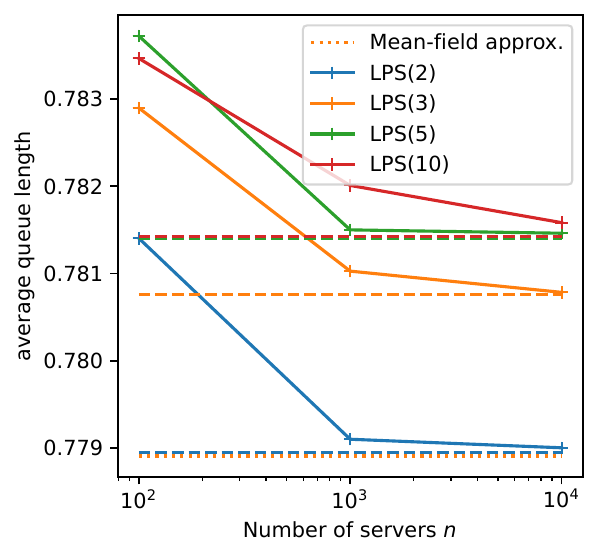}
        \end{tabular}
        &\begin{tabular}{@{}c@{}}
            \includegraphics[width=0.32\linewidth]{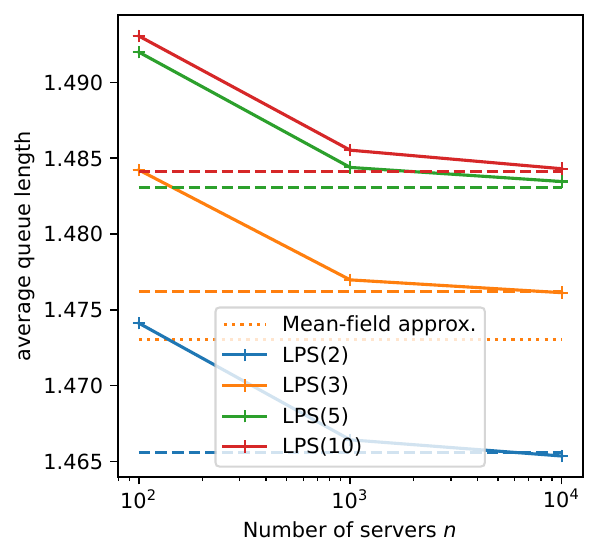}
        \end{tabular}
        &\begin{tabular}{@{}c@{}}
            \includegraphics[width=0.32\linewidth]{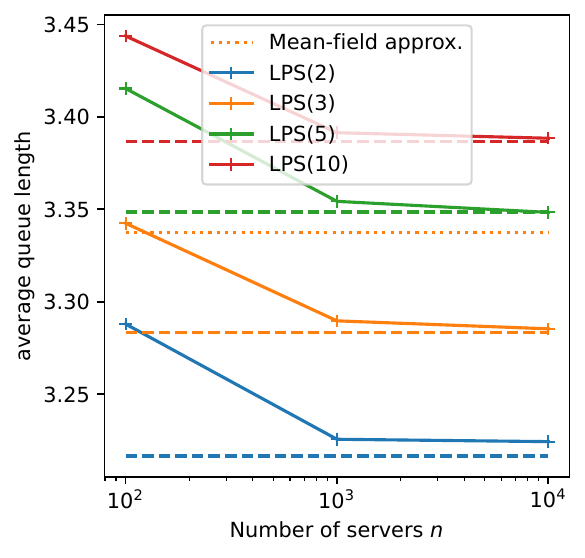}
        \end{tabular}
        \\
        (d) LPS, $\lambda=0.5$
        &(e) LPS, $\lambda=0.7$
        &(f) LPS, $\lambda=0.9$
    \end{tabular}

    \caption{Complement of Figure~\ref{fig:non-asympto-exact} with other values. Average queue length as a function of $n$: we compare the numbers obtained by simulation (for finite $n$) to the different approximations. Solid lines correspond to simulations whereas dashed lines correspond to pair approximation. We observe that all pair-approximations provide estimated queue lengths that are slightly optimistic compared to the values for finite $N$. This is not the case for the mean field approximation that can provide optimistic or pessimistic estimates depending on the scheduling policies or the load. 
    }
    \label{fig:non-asympto-exact2}
\end{figure}

\end{document}